\documentclass[a4paper,reqno]{amsart}
\usepackage{amsmath,amssymb,mathrsfs}
\usepackage{txfonts,bm,pifont}
\usepackage{cite}
\usepackage{float}
\usepackage{stackrel}
\usepackage{subcaption}
\usepackage{graphicx,xcolor}
\usepackage{comment}
\usepackage{enumitem}
\usepackage{longtable}
\usepackage{hhline}
\usepackage {diagbox}
\usepackage{hyperref}

\newtheorem{theorem}{Theorem}[section]

\newtheorem{lemma}[theorem]{Lemma}
\newtheorem{remark}[theorem]{Remark}
\theoremstyle{definition}

\numberwithin{equation}{section}
\numberwithin{figure}{section}
\numberwithin{table}{section}
\newcommand{\wutilde}[1]{\vrule depth 0pt width 0pt%
{\raise0.8pt\hbox{$\smash{{\mathop{#1} \limits_{\displaystyle\widetilde{}}}}$}}}
\newcommand{\wuhat}[1]{\vrule depth 0pt width 0pt%
{\raise0.6pt\hbox{$\smash{{\mathop{#1} \limits_{\displaystyle\widehat{}}}}$}}}

\newcommand{\al}{\alpha}
\newcommand{\be}{\beta}
\newcommand{\de}{\delta}
\newcommand{\ga}{\gamma}

\newcommand{\si}{\sigma}

\newcommand{\la}{\lambda}
\newcommand{\La}{\Lambda}
\newcommand{\ep}{\epsilon}
\newcommand{\ka}{\kappa}
\newcommand{\PDE}{P$\Delta$E}
\newcommand{\ODE}{O$\Delta$E}
\newcommand{\bbZ}{\mathbb{Z}}

\newcommand{\bbC}{\mathbb{C}}

\makeatletter
\long\def\@makecaption#1#2{
 \vskip 10pt
 \setbox\@tempboxa\hbox{#1. #2}
 \ifdim \wd\@tempboxa >\hsize #1. #2\par \else \hbox
to\hsize{\hfil\box\@tempboxa\hfil}
 \fi}
\makeatother
\makeatletter
\@namedef{subjclassname@2020}{%
  \textup{2020} Mathematics Subject Classification}
\makeatother
\begin{document}
\allowdisplaybreaks

\title[]{A higher-order generalization of an $A_4^{(1)}$-surface type $q$-Painlev\'e equation with $\widetilde{W}\left((A_{2N}\rtimes A_1)^{(1)}\times A_1^{(1)}\right)$ symmetry}
\author{Nobutaka Nakazono}
\address{Institute of Engineering, Tokyo University of Agriculture and Technology, 2-24-16 Nakacho Koganei, Tokyo 184-8588, Japan.}
\email{nakazono@go.tuat.ac.jp}
\begin{abstract}
Recently, a birational representation of an extended affine Weyl group of $(A_{2N}\rtimes A_1)^{(1)}$-type, which gives a higher-order generalization of an $A_4^{(1)}$-surface type $q$-Painlev\'e equation, was obtained. In this paper, we extend it to a birational representation of an extended affine Weyl group of $(A_{2N}\rtimes A_1)^{(1)}\times A_1^{(1)}$-type. Moreover, we provide conjectures on periodic reductions from systems of P$\Delta$Es with the CAC property to Painlev\'e type $q$-O$\Delta$Es.
\end{abstract}

\maketitle

\section{Introduction}\label{Introduction}
Fix an integer $N>0$.
In this study, we focus on the symmetry of the following $2N$th-order $q$-{\ODE}\cite{nakazono2023higerA4A6}:
\begin{equation}\label{eqn:intro_dP_odd}
 \text{$q$P$^{(2N)}(A_4^{(1)})$ :\hspace{0.5em}}
 \dfrac{{a_i}^{2N+1}\Big(\overline{F}_iF_i-1\Big)}{{a_i}^{2N+1}-c^{2(-1)^i}F_i}
 =\begin{cases}
 ~ \dfrac{\overline{F}_{i+1}F_{i+1}-1}{1-{a_{i+1}}^{2N+1}c^{2(-1)^i}\overline{F}_{i+1}}
 &\text{if } i=1,\dots,2N-1,\\[1.5em]
 ~ \dfrac{\left(\displaystyle\prod_{k=1}^{2N}{a_k}^k\right)bc}{\displaystyle\prod_{k=1}^N F_{2k-1}}&\text{if } i=2N,
 \end{cases}
\end{equation}
where $b\in\bbC$ is an independent variable, $F_i=F_i(b)\in\bbC$, $i=1,\dots,2N$, are dependent variables, and $a_1,\dots,a_{2N},c,p\in \bbC$ are parameters.
The symbol $\overline{\rule{0em}{0.5em}~\,}$ denotes discrete-time evolution:
\begin{equation}
 \overline{b}=pb,\quad
 \overline{F}_i=F_i(pb),~i=1,\dots,2N,\quad
 \overline{a}_j=a_j,~j=1,\dots,2N,\quad
 \overline{c}=c^{-1},
\end{equation}
where $p\in\bbC$ is a constant.
When $N=1$, the system \eqref{eqn:intro_dP_odd} is equivalent to a $q$-Painlev\'e V equation of $A_4^{(1)}$-surface type \cite{TGCR2004:MR2058894,JNS2016:MR3584386,KN2015:MR3340349,NakazonoN2016:MR3503803}.
(See \cite{nakazono2023higerA4A6} for further details.)

\begin{remark}
For a $q$-difference equation, the symbol ``$q$" is commonly used for its shift parameter.
However, in this paper, the symbol ``$q$" is used in the later arguments, and the relations between the symbol ``$p$" in the system \eqref{eqn:intro_dP_odd} and the symbol ``$q$" in later arguments will be given.
Therefore, to avoid confusion, we use the symbol ``$p$" for the system \eqref{eqn:intro_dP_odd} instead of the symbol ``$q$".
\end{remark}

The following property of the system \eqref{eqn:intro_dP_odd} is known.

\begin{theorem}[\cite{nakazono2023higerA4A6}]\label{theorem:symmetry_A2NA1}
The system \eqref{eqn:intro_dP_odd} has an extended affine Weyl group symmetry of $(A_{2N}\rtimes A_1)^{(1)}$-type, that is,
it can be derived from a birational representation of an extended affine Weyl group of $(A_{2N}\rtimes A_1)^{(1)}$-type, denoted by $\widetilde{W}\left((A_{2N}\rtimes A_1)^{(1)}\right)$.
\end{theorem}

\begin{remark}\label{remark:symmetry_A2NA1}
An affine Weyl group of $(A_{2N}\rtimes A_1)^{(1)}$-type refers to the semi-direct product of an affine Weyl group of $A_{2N}^{(1)}$-type and that of $A_1^{(1)}$-type.
The extended affine Weyl group of $(A_{2N}\rtimes A_1)^{(1)}$-type in Theorem \ref{theorem:symmetry_A2NA1} indicates an affine Weyl group of $(A_{2N}\rtimes A_1)^{(1)}$-type extended by an automorphism of Dynkin diagrams.
For more details, see \S \ref{subsection:A2NA1_review}.
In particular, refer to \eqref{eqns:fundamental_A2NA1} for the fundamental relations of: 
\begin{equation}
 \widetilde{W}\left((A_{2N}\rtimes A_1)^{(1)}\right)
 =\Big(\langle s_0,\dots,s_{2N}\rangle\rtimes\langle w_0,w_1\rangle\Big)\rtimes\langle r\rangle.
\end{equation}
\end{remark}

It is known that Painlev\'e-type {\ODE}s are derived from birational actions of the infinite order elements of (extended) affine Weyl groups.
For instance, considering the transformation $T\in\widetilde{W}\left((A_{2N}\rtimes A_1)^{(1)}\right)$ acting on the parameters $\{a_1,\dots,a_{2N},b,c,p\}$ such that 
\begin{equation}\label{eqn:intro_T}
 T(a_i)=a_i,~i=1,\dots,2N,\quad
 T(b)=pb,\quad
 T(c)=c^{-1},\quad
 T(p)=p,
\end{equation}
as a time evolution, we obtain the system \eqref{eqn:intro_dP_odd}. 
(See Remark \ref{remark:T_hT} for transformation $T$.)
In this context, parameter $b$ can be regarded as the independent variable of the system \eqref{eqn:intro_dP_odd}.
The group $\widetilde{W}\left((A_{2N}\rtimes A_1)^{(1)}\right)$ includes transformations that regard not only the parameter $b$ but also the parameter $a_i$ as independent variables.
For instance, there exists a transformation $\hat{T}\in\widetilde{W}\left((A_{2N}\rtimes A_1)^{(1)}\right)$ that satisfies the following action on the parameters:
\begin{equation}\label{eqn:intro_hT}
 \hat{T}(a_1)=pa_1,\quad
 \hat{T}(a_i)=a_i,~i=2,\dots,2N,\quad
 \hat{T}(b)=b,\quad
 \hat{T}(c)=c,\quad
 \hat{T}(p)=p.
\end{equation}
(See Remark \ref{remark:T_hT} for the transformation $\hat{T}$.)
However, focusing on the action on the parameter $c$, it is observed that $\widetilde{W}\left((A_{2N}\rtimes A_1)^{(1)}\right)$ contains only transformations acting as $c\to c$ or $c\to c^{-1}$.
This is attributed to the fact that in the study of \cite{nakazono2023higerA4A6}, the parameter $c$ is not an essential parameter derived from the system of {\PDE}s considering reduction but rather a parameter that can be set to $1$ through a gauge transformation before reduction.

The purpose of this study is to extend the transformation group $\widetilde{W}\left((A_{2N}\rtimes A_1)^{(1)}\right)$ to include a transformation $\widetilde{T}$ that regards the parameter $c$ as an independent variable of an {\ODE}.
To achieve this, we consider the reduction of a different system of {\PDE}s from that treated in \cite{nakazono2023higerA4A6}.
(See Remark \ref{remark:wT} for the transformation $\widetilde{T}$.)

The main theorem of this paper is the following.

\begin{theorem}\label{theorem:main}
The system \eqref{eqn:intro_dP_odd} has an extended affine Weyl group symmetry of $(A_{2N}\rtimes A_1)^{(1)}\times A_1^{(1)}$-type, that is,
it can be derived from a birational representation of an extended affine Weyl group of $(A_{2N}\rtimes A_1)^{(1)}\times A_1^{(1)}$-type, denoted by $\widetilde{W}\left((A_{2N}\rtimes A_1)^{(1)}\times A_1^{(1)}\right)$.
\end{theorem}
\begin{proof}
This theorem is proven in \S \ref{subsection:proof_mainThm}.
\end{proof}

\begin{remark}\label{remark:symmetry_A2NA1A1}
An affine Weyl group of $(A_{2N}\rtimes A_1)^{(1)}\times A_1^{(1)}$-type means the direct product of an affine Weyl group of $(A_{2N}\rtimes A_1)^{(1)}$-type and that of $A_1^{(1)}$-type.
The extended affine Weyl group of $(A_{2N}\rtimes A_1)^{(1)}\times A_1^{(1)}$-type in Theorem \ref{theorem:main} indicates an affine Weyl group of $(A_{2N}\rtimes A_1)^{(1)}\times A_1^{(1)}$-type extended by an automorphism of Dynkin diagrams.
For details, see \S \ref{subsection:proof_mainThm}.
In particular, see \eqref{eqns:fundamental_A2NA1} and \eqref{eqn:fundamental_mu} for the fundamental relations of 
\begin{equation}
 \widetilde{W}\left((A_{2N}\rtimes A_1)^{(1)}\times A_1^{(1)}\right)
 =\bigg(\Big(\langle s_0,\dots,s_{2N}\rangle\rtimes\langle w_0,w_1\rangle\Big)\times\langle \mu_0,\mu_1\rangle\bigg)\rtimes\langle r\rangle.
\end{equation}
\end{remark}

\subsection{Background}\label{subsection:Background}
\subsubsection*{\bf Painlev\'e equations}
In the early 20th century, to find a new class of special functions, Painlev\'e and Gambier classified all the ordinary differential equations of the type 
\begin{equation}
 y''=F(y',y,t),
\end{equation}
where $y=y(t)$, $'=d/dt$ and $F$ is a function meromorphic in $t$ and rational in $y$ and $y'$, 
with the Painlev\'e property (solutions do not have movable singularities other than poles)\cite{PainleveP1900:zbMATH02665472,PainleveP1902:MR1554937,PainleveP1907:zbMATH02647172,GambierB1910:MR1555055}.
As a result, they obtained six new equations that are collectively referred to as Painlev\'e equations.
Note that the Painlev\'e VI equation was obtained by Fuchs \cite{FuchsR1905:quelques} before Painlev\'e and Gambier.

After the discovery, the Painlev\'e equations withdrew from the stage of ``modern mathematics" for a while.
The Painlev\'e equations regained attention after the 1970s because they appeared in mathematical physics research.
For instance, solutions of the Painlev\'e equations (Painlev\'e transcendents) were rediscovered as scaling functions for the two-dimensional Ising model on a square lattice \cite{WMTB1976:PhysRevB.13.316} and as similarity solutions of the soliton equations describing specific shallow water waves (solitons) \cite{AS1977:PhysRevLett.38.1103}.

In the late 20th century, Okamoto described a geometric framework (Okamoto's space of initial values) to study Painlev\'e equations \cite{OkamotoK1979:MR614694,OKSO2006:MR2277519}.
Consequently, the Painlev\'e III equation is classified into three types, which in turn classifies the Painlev\'e equations into eight types.

\subsubsection*{\bf $q$-Painlev\'e equations}
$q$-Painlev\'e equations are a family of second-order nonlinear $q$-{\ODE}s.
Historically, they were obtained as $q$-discrete analogues of the Painlev\'e equations (see, for example, \cite{GRP1991:MR1125950,JS1996:MR1403067}).
Similar to the Painlev\'e equations, $q$-Painlev\'e equations are known to describe special solutions of various discrete soliton equations \cite{GRSWC2005:MR2117991,HHJN2007:MR2303490,FJN2008:MR2425981,OrmerodCM2012:MR2997166,
HHNS2015:MR3317164}.
In particular, famous is imposing periodic conditions on the discrete soliton equations\cite{HayM2007:MR2371129,HHJN2007:MR2303490,OVHQ2014:MR3215839,OrmerodCM2012:MR2997166,OVQ2013:MR3030178,KNY2002:MR1958118,TsudaT2010:MR2563787,PNGR1992:MR1162062,JNS2016:MR3584386,JNS2015:MR3403054,JNS2014:MR3291391,nakazono2023higerA4A6,AJT2020:zbMATH07212273}.

There are six (or eight) Painlev\'e equations; however, an infinite number of $q$-Painlev\'e equations are known to exist.
By considering Sakai's space of initial values \cite{SakaiH2001:MR1882403}, which is an extension of Okamoto's space of initial values, $q$-Painlev\'e equations can be classified into nine surface types (see Figure \ref{fig:Sakai_surface}).
In this study, a $q$-Painlev\'e equation of $X$-surface type refers to one $q$-Painlev\'e equation belonging to Sakai's space of initial values of type $X$. 
Note that infinitely many $q$-Painlev\'e equations exist for the same surface type.

\begin{figure}[htbp]
\[
 A_0^{(1)\ast}
 \to A_1^{(1)}
 \to A_2^{(1)}
 \to A_3^{(1)}
 \to A_4^{(1)}
 \to A_5^{(1)}
 \to A_6^{(1)}~
 \begin{matrix}
 \nearrow\\
 \searrow
 \end{matrix}
 \begin{matrix}
 ~A_7^{(1)}\\[1.6em]
 ~{A_7^{(1)}}'
 \end{matrix}
\]
\caption{
Types of Sakai's spaces of initial values for $q$-Painlev\'e equations.
The surface degenerates in the direction of the arrow due to the specialization and confluence of the base points that characterize the surface types.
}
\label{fig:Sakai_surface}
\end{figure}

\subsubsection*{\bf (Extended) affine Weyl group symmetry}
A B\"acklund transformation is a transformation from an integrable system to another, possibly the same, integrable system. 
In the case of a Painlev\'e-type differential/difference equation, auto-B\"acklund transformations exist, which are transformations that map an equation to itself with varying parameters. 
These collectively form an (extended) affine Weyl group.
(See, for example \cite{KNY2017:MR3609039,HJN2016:MR3587455,SakaiH2001:MR1882403,book_NoumiM2004:MR2044201}.)
In this sense, the equation is said to have (extended) affine Weyl group symmetry.
Note that an (extended) affine Weyl group symmetry of a Painlev\'e-type {\ODE} does not always refer solely to its B\"acklund transformations, but it often refers to a large group of transformations, including its time evolution.
For example, see Theorems \ref{theorem:symmetry_A2NA1} and \ref{theorem:main}.

\subsubsection*{\bf (Extended) KNY's representation}
In \cite{KNY2002:MR1958118}, Kajiwara-Noumi-Yamada showed a birational representation of an extended affine Weyl group of $(A_{m-1}\times A_{n-1})^{(1)}$-type (KNY's representation), where $m$ and $n$ are integers greater than or equal to $2$, except for $(m,n)=(2,2)$.
Note that KNY's representation is essentially the same even if $m$ and $n$ are interchanged \cite{NY2018:zbMATH06876428}.
Recently, it has been reported that KNY's representation can be extended to a birational representation of an extended affine Weyl group of $(A_{m-1}\times A_{n-1}\times A_{g-1})^{(1)}$-type (extended KNY's representation), where $g$ is the common greatest divisor of $m$ and $n$ \cite{MOT2018:AmAnAg,MOT2023:AmAnAgArxiv}.

It is well known that the $(A_L\times A_1)^{(1)}$-type KNY's representation, where $L\in\bbZ_{\geq 2}$, yields Painlev\'e-type $q$-{\ODE}s, including $q$-Painlev\'e equations as second-order $q$-{\ODE}s.
Indeed, the $(A_2\times A_1)^{(1)}$-type KNY's representation gives $q$-Painlev\'e equations of $A_5^{(1)}$-surface type \cite{KNY2002:MR1917133}, and the $(A_3\times A_1)^{(1)}$-type gives $q$-Painlev\'e equations of $A_3^{(1)}$-surface type \cite{KNY2002:MR1917133,TakenawaT2003:MR1996297}.
As mentioned above, the $(A_{2N+1}\times A_1)^{(1)}$-type KNY's representation, where $N\in\bbZ_{\geq 1}$, can be extended to the $(A_{2N+1}\times A_1\times A_1)^{(1)}$-type extended KNY's representation.
In \cite{OS2020:10.1093/imrn/rnaa283}, the explicit forms of the Painlev\'e-type $q$-{\ODE}s were obtained from the $(A_{2N+1}\times A_1\times A_1)^{(1)}$-type extended KNY's representation.

\subsubsection*{\bf Motivation for this study}
In \cite{nakazono2023higerA4A6}, a birational representation of an extended affine Weyl group of $(A_{2N}\rtimes A_1)^{(1)}$-type, where $N\in\bbZ_{\geq 1}$,  was constructed, and the system \eqref{eqn:intro_dP_odd} was derived from its action. 
This representation is presumed to be a degenerate version of KNY's representation (see Appendix \ref{section:conjecture}); therefore, similar to KNY's representation, it is believed to be extendable. 
This is the motivation for the present study.
It should be noted that the methods used to extend them differ. 
In this study, the extension was explored using the theory of consistency around the cube (CAC) property (see \cite{BS2002:MR1890049,NijhoffFW2002:MR1912127,WalkerAJ:thesis,NW2001:MR1869690} for the CAC property), whereas the extension of KNY's representation was performed using the theory of cluster algebra (see \cite{BGM2018:zbMATH06868180,HI2014:zbMATH06381200,OkuboN2015:zbMATH06499648,IIKNS2010:zbMATH05704436,NobeA2016:zbMATH06618416} for the theory of cluster algebra for discrete integrable systems).

\subsection{Notation and Terminology}\label{subsection:notation_definitions}
This paper will use the following notations and terminologies for conciseness.
\begin{itemize}
\item 
An ordinary difference equation is written as {\ODE} and a partial difference equation is written as {\PDE}. 
In particular, an ordinary multiplicative-type difference ($q$-difference) equation is expressed as {$q$-\ODE}.
\item 
For transformations $s$ and $r$, the symbol $sr$ means the composite transformation $s\circ r$.
\item 
In the context of transformations, the ``$1$" signifies the identity transformation.
\item
For transformation $s$, the relation $s^\infty=1$ implies that there is no positive integer $k$ such as
$s^k=1$.
\item 
If the subscript number is greater than the superscript number in the product symbol, $1$ is assumed.
For example,
\begin{equation}
 \prod_{k=1}^0 2^k=1.
\end{equation}
\item 
If the subscript number is greater than the superscript number in the summation symbol, $0$ is assumed.
For example,
\begin{equation}
 \sum_{k=1}^0 2^k=0.
\end{equation}
\item 
When an equation number has a subscript such as "$l=0$", it signifies the equation obtained by substituting $l=0$ into the equation corresponding to that equation number.
For example, with reference to 
\begin{equation}\label{eqn:example1}
 a_l+a_{l+1}+a_{l+2}=0,
\end{equation}
the equation \eqref{eqn:example1}$_{l=0}$ is equivalent to 
\begin{equation}
 a_0+a_1+a_2=0.
\end{equation}
The same applies to the case where ``$l\to l+1$" or other symbols are in the subscripts.
For instance, the equation \eqref{eqn:example1}$_{l\to l+1}$ is equivalent to
\begin{equation}
 a_{l+1}+a_{l+2}+a_{l+3}=0.
\end{equation}
\end{itemize}
\subsection{Outline of the paper}
This paper is organized as follows.
In \S \ref{section:transformations_si_w_mu}, we introduce system \eqref{eqns:PDEs_UV} and the transformations $\{\si,w,\mu\}$ that remain the system invariant.
Subsequently, by imposing periodic conditions on the system \eqref{eqns:PDEs_UV}, we obtain the birational actions of transformations $\{\si,w,\mu\}$.
In \S \ref{section:proof_mainThm}, using the transformation group $\widetilde{W}\left((A_{2N}\rtimes A_1)^{(1)}\right)$ given in \cite{nakazono2023higerA4A6} and transformations $\{\si,w,\mu\}$, we prove Theorem \ref{theorem:main}.
Some concluding remarks are given in \S \ref{ConcludingRemarks}.
In Appendix \ref{section:system_CAC_proof}, we show that system \eqref{eqns:PDEs_UV} has the CAC and tetrahedron properties.
In Appendix \ref{section:conjecture}, we provide conjectures on periodic reductions from systems of {\PDE}s with the CAC property to Painlev\'e type $q$-{\ODE}s.
\section{Transformations $\si$, $w$ and $\mu$}\label{section:transformations_si_w_mu} 
In this section, we define the transformations $\{\si,w,\mu\}$ and obtain their birational actions necessary for the proof of Theorem \ref{theorem:main} through a staircase reduction of a system of {\PDE}s. 

\subsection{A system of {\PDE}s with the CAC property}\label{subsection:system_CAC}
Let $(l,m)\in\bbZ^2$.
We focus on the following system of {\PDE}s for $U_{l,m},V_{l,m}\in\bbC$:
\begin{subequations}\label{eqns:PDEs_UV}
\begin{align}
 &\dfrac{U_{l+1,m+1}}{U_{l,m}}+{\ga_{l+m}}^4\dfrac{U_{l,m+1}}{U_{l+1,m}}+\al_l\be_m\ga_{l+m}=0,\label{eqn:PDE_UV_A}\\
 &\dfrac{V_{l+1,m+1}}{V_{l,m}}+{\de_{l+m}}^4\dfrac{V_{l,m+1}}{V_{l+1,m}}+\al_l\be_m\de_{l+m}=0,\label{eqn:PDE_UV_AA}\\
 &\dfrac{U_{l,m+1}}{V_{l,m}}
 =\begin{cases}
 \dfrac{V_{l,m+1}}{U_{l,m}}&\text{(if $(l+m)$ is even)},\\[1em]
 \dfrac{\ga^4}{\de^4}\left(\dfrac{V_{l,m+1}}{U_{l,m}}+\dfrac{\be_m\de^4(\ga-\de)}{\ep^{l-2m-1}}\right)&\text{(if $(l+m)$ is odd)},
 \end{cases}\label{eqn:PDE_UV_B}\\
 &\dfrac{U_{l,m}}{V_{l+1,m}}
 =\begin{cases}
 \dfrac{\ga^3}{\de^3}\left(\dfrac{V_{l,m}}{U_{l+1,m}}-\dfrac{\de^3(\ga-\de)}{\al_l\ep^{l-2m}}\right)&\text{(if $(l+m)$ is even)},\\[1em]
 \dfrac{\ga}{\de}\left(\dfrac{U_{l+1,m}}{V_{l,m}}+\dfrac{\ga(\ga-\de)}{\al_l\ep^{l-2m-1}}\right)^{-1}&\text{(if $(l+m)$ is odd)},
 \end{cases}\label{eqn:PDE_UV_C}
\end{align}
\end{subequations}
where $\{\al_k,\be_k,\ga_k,\de_k,\ga,\de,\ep\}_{k\in\bbZ}$ are complex parameters satisfying
\begin{equation}
 \ga_l=\begin{cases}
 \ga&\text{(if $l$ is even)},\\[0.5em]
 \ga^{-1}&\text{(if $l$ is odd)},
 \end{cases}\quad
 \de_m=\begin{cases}
 \de&\text{(if $m$ is even)},\\[0.5em]
 \de^{-1}&\text{(if $m$ is odd)},
 \end{cases}\quad
 \ep=\ga\de.
\end{equation}
The system \eqref{eqns:PDEs_UV} has the CAC and tetrahedron properties.
Further details are provided in Appendix \ref{section:system_CAC_proof}.

The transformations $\{\si,w,\mu\}$ are defined by the actions on the parameters $\{\al_k,\be_k,\ga_k,\de_k,\ga,\de,\ep\}_{k\in\bbZ}$ and the variables $\{U_{l,m},V_{l,m}\}$ as follows.
\begin{subequations}\label{eqns:def_sigma_w_mu}
\begin{align}
 &\si(\al_k)=\ep^2\al_{k+2},\quad
 \si(\be_k)=\ep^{-2}\be_k,\quad
 \si(\ga_k)=\ga_k,\quad
 \si(\de_k)=\de_k,\quad
 \si(\ga)=\ga,\notag\\
 &\si(\de)=\de,\quad
 \si(\ep)=\ep,\quad
 \si(U_{l,m})=U_{l+2,m},\quad
 \si(V_{l,m})=V_{l+2,m},\\
 &w(\al_k)=\ep^{2N-4}\al_{-k+2N},\quad
 w(\be_k)=\ep^{-2N+4}\be_{-k-1},\quad
 w(\ga_k)={\ga_k}^{-1},\quad
 w(\de_k)={\de_k}^{-1},\notag\\
 &w(\ga)={\ga}^{-1},\quad
 w(\de)={\de}^{-1},\quad
 w(\ep)={\ep}^{-1},\quad
 w(U_{l,m})=\dfrac{1}{U_{-l+2N+1,-m}},\notag\\
 &w(V_{l,m})=\dfrac{1}{V_{-l+2N+1,-m}},\\
 &\mu(\al_k)=\al_k,\quad
 \mu(\be_k)=\be_k,\quad
 \mu(\ga_k)=\de_k,\quad
 \mu(\de_k)=\ga_k,\quad
 \mu(\ga)=\de,\quad
 \mu(\de)=\ga,\notag\\
 &\mu(\ep)=\ep,\quad
 \mu(U_{l,m})=V_{l,m},\quad
 \mu(V_{l,m})=U_{l,m}.
\end{align}
\end{subequations}
We can easily verify that system \eqref{eqns:PDEs_UV} is invariant under the actions of transformations $\{\si,w,\mu\}$ and that these transformations satisfy the following relations:
\begin{equation}\label{eqn:relations_sigma_w_mu}
 \si^{\,\infty}=w^2=\mu^2=1,\quad
 \si w=w\si^{-1},\quad
 \si\mu=\mu\si,\quad
 w\mu=\mu w.
\end{equation}

\subsection{A staircase reduction of the system \eqref{eqns:PDEs_UV}}\label{subsection:staircase_reduction}
Fix an integer $N>0$.
We consider the $(2N+1,1)$-periodic conditions
\begin{equation}\label{eqn:(2N+1,1)_periodic_cond_1}
 U_{l+2N+1,m+1}=U_{l,m},\quad
 V_{l+2N+1,m+1}=V_{l,m},
\end{equation}
which are equivalent to expressing $U_{l,m}$ and $V_{l,m}$ as
\begin{equation}\label{eqn:(2N+1,1)_periodic_cond_2}
 U_{l,m}=U(l-(2N+1)m),\quad
 V_{l,m}=V(l-(2N+1)m).
\end{equation}
By substituting \eqref{eqn:(2N+1,1)_periodic_cond_2} into \eqref{eqns:PDEs_UV}, we obtain the following system of {\ODE}s:
\begin{subequations}\label{eqns:ODEs_UV}
\begin{align}
 &\dfrac{U(l)}{U(l+2N+2)}+{\ga_l}^4\dfrac{U(l+1)}{U(l+2N+1)}+\dfrac{\al_l \be_0{\ga_l}^3}{\ep^{2N-1}}=0,\label{eqn:ODE_UV_A}\\
 &\dfrac{V(l)}{V(l+2N+2)}+{\de_l}^4\dfrac{V(l+1)}{V(l+2N+1)}+\dfrac{\al_l \be_0{\de_l}^3}{\ep^{2N-1}}=0,\label{eqn:ODE_UV_AA}\\
 &\dfrac{U(2l+1)}{V(2l+2N+2)}-\dfrac{V(2l+1)}{U(2l+2N+2)}=0,\label{eqn:ODE_UV_Be}\\
 &\dfrac{U(2l)}{V(2l+2N+1)}-\dfrac{\ga^4V(2l)}{\de^4U(2l+2N+1)}-\dfrac{\be_0 \ga^4(\ga-\de)}{\ep^{2(l+N)}}=0,\label{eqn:ODE_UV_Bo}\\
 &\dfrac{U(2l)}{V(2l+1)}-\dfrac{\ga^3V(2l)}{\de^3U(2l+1)}+\dfrac{\ga^3(\ga-\de)}{\ep^{2l}\al_{2l}}=0,\label{eqn:ODE_UV_Ce}\\
 &\dfrac{U(2l+2)}{V(2l+1)}-\dfrac{\ga V(2l+2)}{\de U(2l+1)}+\dfrac{\ga(\ga-\de)}{\ep^{2l}\al_{2l+1}}=0,\label{eqn:ODE_UV_Co}
\end{align}
\end{subequations}
together with the conditions for the parameters
\begin{equation}\label{eqn:(2N+1,1)_periodic_cond_3}
 \al_{l+2N+1}=\ep^{-2N+1}\al_l,\quad
 \be_{m+1}=\ep^{2N-1}\be_m.
\end{equation}
Define the parameters $\{a_0,\dots,a_{2N},b,c,q\}$ and the variables $\{f(l),g(l)\}$ by
\begin{subequations}\label{eqns:def_abcqfg}
\begin{align}
 &a_0=\ep^{\frac{1-2N}{2N+1}}\left(\dfrac{\al_0}{\al_{2N}}\right)^{\frac{1}{2N+1}},\quad
 a_i=\left(\dfrac{\al_i}{\al_{i-1}}\right)^{\frac{1}{2N+1}},~ i=1,\dots,2N,\quad
 b=\left(\prod_{k=0}^{2N} \al_{k}\right)^{\frac{1}{2N+1}}\be_0,\\
 &c=\ga,\quad
 q=\ep^{\frac{1-2N}{2N+1}},\quad
 f(l)=\dfrac{U(l-1)}{U(l+1)},\quad
 g(l)=\dfrac{V(l-1)}{V(l+1)},
\end{align}
\end{subequations}
where the following holds:
\begin{equation}\label{eqn:cond_ai_lattice}
 \prod_{i=0}^{2N}a_i=q.
\end{equation}
Then, the following lemma holds.

\begin{lemma}\label{lemma:f(l)_g(l)}
The following $2N$th-order $q$-{\ODE}s for $f(l)$ and $g(l)$ hold.
\begin{subequations}
\begin{align}
 &f(l+2N+1)f(l+1)\left(\prod_{k=1}^{N-1}f(l+2k+1)\right)
 +{\ga_l}^4\left(\prod_{k=1}^{N}f(l+2k)\right)\notag\\
 &=-\dfrac{q^{l+2N}a_l b {\ga_l}^3}{\displaystyle\prod_{k=1}^{2N-1}{a_{l+k}}^{2N-k}},\label{eqn:red_f_ODE}\\
 &g(l+2N+1)g(l+1)\left(\prod_{k=1}^{N-1}g(l+2k+1)\right)
 +{\de_l}^4\left(\prod_{k=1}^{N}g(l+2k)\right)\notag\\
 &=-\dfrac{q^{l+2N}a_l b {\de_l}^3}{\displaystyle\prod_{k=1}^{2N-1}{a_{l+k}}^{2N-k}},\label{eqn:red_g_ODE}
\end{align}
\end{subequations}
where $a_{i+2N+1}=a_i$ for arbitrary $i\in\bbZ$.
Here, $\ga_l$ and $\de_l$ are expressed using $c$ and $q$, respectively, as follows:
\begin{equation}
 \ga_l=\begin{cases}
 c&\text{(if $l$ is even)},\\[1em]
 \dfrac{1}{c}&\text{(if $l$ is odd)},
 \end{cases}\qquad
 \de_l=\begin{cases}
 \dfrac{\ep}{c}&\text{(if $l$ is even)},\\[1em]
 \dfrac{c}{\ep}&\text{(if $l$ is odd)},
 \end{cases}
\end{equation}
where $\ep=q^{\frac{2N+1}{1-2N}}$.
\end{lemma}
\begin{proof}
Equations \eqref{eqn:red_f_ODE} and \eqref{eqn:red_g_ODE} are obtained from Equations \eqref{eqn:ODE_UV_A} and \eqref{eqn:ODE_UV_AA}, respectively.
\end{proof}

Define 
\begin{equation}\label{eqn:def_fi_gj}
 f_i=f(i),\quad
 g_i=g(i),\quad i=1,\dots,2N.
\end{equation}
Then, the following lemma holds.

\begin{lemma}\label{lemma:g_f}
The following relations hold:
\begin{subequations}\label{eqns:g_f}
\begin{align}
 &\left(1-\dfrac{f_1}{{a_1}^{2N+1}c^{2}}\right)
 \dfrac{f_1-\ep^{2N-3}c^{6-4N}g_1}{f_1-q^{2N+1}\ep^{2N-3}c^4g_1}\notag\\
 &\hspace{1em}+\sum_{k=1}^{N-1}\dfrac{\displaystyle\prod_{i=1}^k f_{2i-1}}{\left(\displaystyle\prod_{i=1}^{2k}{a_i}^{2N+1}\right)c^{4k} \left(\displaystyle\prod_{i=1}^k f_{2i}\right)}\left(1-\dfrac{f_{2k+1}}{{a_{2k+1}}^{2N+1}c^2}\right)\notag\\
 &\hspace{1em}+\dfrac{{a_0}^{2N+1}}{q^{2N+1}c^{4N}\left(\displaystyle\prod_{i=1}^N f_{2i}\right)}\left(\rule{0em}{2.5em}\right.\dfrac{q^{2N}bc}{\displaystyle\prod_{i=0}^{2N-1}{a_i}^{2N-i}}+\displaystyle\prod_{i=1}^N f_{2i-1}\left.\rule{0em}{2.5em}\right)
 =0,\label{eqn:g1_f}\\
 &\dfrac{g_{2l_1+3}\left(\ep^2f_{2l_1+1}-c^4g_{2l_1+1}\right)\left({a_{2l_1+3}}^{2N+1}c^2-f_{2l_1+3}\right)}{f_{2l_1+2}\left(\ep^2f_{2l_1+3}-c^4g_{2l_1+3}\right)\left({a_{2l_1+1}}^{2N+1}c^2-f_{2l_1+1}\right)}
 =\ep^2{a_{2l_1+2}}^{2N+1}{a_{2l_1+3}}^{2N+1},\label{eqn:g2l+3_f}\\
 &\dfrac{\left(c^2-\ep^2{a_{2l_2+2}}^{2N+1}g_{2l_2+2}\right)\left(\ep^2{a_{2l_2+1}}^{2N+1}-c^2g_{2l_2+1}\right)}{\left(1-{a_{2l_2+2}}^{2N+1}c^2f_{2l_2+2}\right)\left({a_{2l_2+1}}^{2N+1}c^2-f_{2l_2+1}\right)}=\ep^2,\label{eqn:g2l+2_f}
\end{align}
\end{subequations}
where 
\begin{equation}
 l_1=0,\dots,N-2,\quad
 l_2=0,\dots,N-1,\quad
 \ep=q^{\frac{2N+1}{1-2N}}.
\end{equation}
\end{lemma}
\begin{proof}
For simplicity, let
\begin{equation}
 A(l)=\dfrac{\ep^{2l-3}V(2l)}{\ga^{4l-3}(\ga-\de)U(2l)},
\end{equation}
which satisfies
\begin{equation}\label{eqn:A_fg}
 \dfrac{A(l)}{A(l+1)}=\dfrac{\ga^4g(2l+1)}{\ep^2f(2l+1)}.
\end{equation}
Eliminating $V(2l+1)$ from Equations \eqref{eqn:ODE_UV_Ce} and \eqref{eqn:ODE_UV_Co}, we obtain
\begin{equation}\label{eqn:A_U_1}
 A(l+1)-A(l)=-\dfrac{U(2l+1)}{\al_{2l}\ga^{4l}U(2l)}\left(1-\dfrac{\al_{2l}U(2l)}{\al_{2l+1}\ga^{2}U(2l+2)}\right).
\end{equation}
Then, from Equation \eqref{eqn:A_U_1}, we obtain the following relations:
\begin{align}
 &A(N)=A(l+1)
 -\sum_{k=l+1}^{N-1}\dfrac{U(2k+1)}{\al_{2k}\ga^{4k} U(2k)}\left(1-\dfrac{\al_{2k}U(2k)}{\al_{2k+1}\ga^{2}U(2k+2)}\right),\label{eqn:A_U_1_1}\\
 &A(0)=A(l)
 +\sum_{k=0}^{l-1}\dfrac{U(2k+1)}{\al_{2k}\ga^{4k}U(2k)}\left(1-\dfrac{\al_{2k}U(2k)}{\al_{2k+1}\ga^{2}U(2k+2)}\right),\label{eqn:A_U_1_2}
\end{align}
where $0\leq l\leq N-1$.
Eliminating $V(2N+1)$ from Equations \eqref{eqn:ODE_UV_Bo}$_{l=0}$ and \eqref{eqn:ODE_UV_Ce}$_{l=N}$, we obtain
\begin{equation}\label{eqn:A_U_1_3}
 A(N)-\dfrac{\ep^{2N-1}}{\ga^{4N-2}}A(0)
 =\dfrac{U(2N+1)}{\al_{2N}\ga^{4N}U(0)}\left(\al_{2N}\be_0\ga+\dfrac{U(0)}{U(2N)}\right).
\end{equation}
By eliminating $A(0)$ and $A(N)$ from Equations \eqref{eqn:A_U_1_1}, \eqref{eqn:A_U_1_2}, and \eqref{eqn:A_U_1_3}, we obtain
\begin{align}\label{eqn:A_U_2}
 A(l+1)-\dfrac{\ep^{2N-1}}{\ga^{4N-2}}A(l)
 =&\sum_{k=l+1}^{N-1}\dfrac{U(2k+1)}{\al_{2k}\ga^{4k} U(2k)}\left(1-\dfrac{\al_{2k}U(2k)}{\al_{2k+1}\ga^{2}U(2k+2)}\right)\notag\\
 &+\dfrac{\ep^{2N-1}}{\ga^{4N-2}}\sum_{k=0}^{l-1}\dfrac{U(2k+1)}{\al_{2k}\ga^{4k}U(2k)}\left(1-\dfrac{\al_{2k}U(2k)}{\al_{2k+1}\ga^{2}U(2k+2)}\right)\notag\\
 &+\dfrac{U(2N+1)}{\al_{2N}\ga^{4N}U(0)}\left(\al_{2N}\be_0\ga+\dfrac{U(0)}{U(2N)}\right).
\end{align}
From \eqref{eqn:A_U_1} and \eqref{eqn:A_U_2}, we obtain
\begin{align}\label{eqn:A_U_3}
 &\left(1-\dfrac{\al_{2l}U(2l)}{\al_{2l+1}\ga^{2}U(2l+2)}\right)
 \dfrac{A(l+1)-\dfrac{\ep^{2N-1}}{\ga^{4N-2}}A(l)}{A(l+1)-A(l)}\notag\\
 &=-\sum_{k=l+1}^{N-1}\dfrac{\al_{2l} U(2k+1)U(2l)}{\al_{2k}\ga^{4(k-l)} U(2l+1)U(2k)}\left(1-\dfrac{\al_{2k}U(2k)}{\al_{2k+1}\ga^{2}U(2k+2)}\right)\notag\\
 &\hspace{1em}-\dfrac{\ep^{2N-1}}{\ga^{4N-2}}\sum_{k=0}^{l-1}\dfrac{\al_{2l}\ga^{4(l-k)}U(2k+1)U(2l)}{\al_{2k}U(2l+1)U(2k)}\left(1-\dfrac{\al_{2k}U(2k)}{\al_{2k+1}\ga^{2}U(2k+2)}\right)\notag\\
 &\hspace{1em}-\dfrac{\al_{2l}U(2N+1)U(2l)}{\al_{2N}\ga^{4(N-l)}U(2l+1)U(0)}\left(\al_{2N}\be_0\ga+\dfrac{U(0)}{U(2N)}\right),
\end{align}
where $0\leq l\leq N-1$.
Equation \eqref{eqn:A_U_3}$_{l=0}$ gives \eqref{eqn:g1_f}.

From Equations \eqref{eqn:A_U_3} and \eqref{eqn:A_U_3}$_{l\to l+1}$, we obtain
\begin{equation}
 \dfrac{A(l+1)-A(l)}{A(l+2)-A(l+1)}
 =\dfrac{\al_{2l+2}\ga^4U(2l+2)U(2l+1)}{\al_{2l}U(2l)U(2l+3)}
 \left(\cfrac{1-\dfrac{\al_{2l}U(2l)}{\al_{2l+1}\ga^2U(2l+2)}}{1-\dfrac{\al_{2l+2}U(2l+2)}{\al_{2l+3}\ga^2U(2l+4)}}\right),
\end{equation}
where $0\leq l\leq N-2$, which gives Equation \eqref{eqn:g2l+3_f}.

By solving Equation \eqref{eqn:ODE_UV_Ce} with $V(2l)$ and Equation \eqref{eqn:ODE_UV_Ce}$_{l\to l+1}$ with $V(2l+3)$, and substituting the results into Equation \eqref{eqn:g2l+2_f}, we obtain
\begin{align}
 &\left(\dfrac{U(2l+2)}{V(2l+1)}-\dfrac{\ga V(2l+2)}{\de U(2l+1)}+\dfrac{\ga(\ga-\de)}{\ep^{2l}\al_{2l+1}}\right)\notag\\
 &\times\left(\dfrac{U(2l)}{V(2l+1)}-\dfrac{\al_{2l+2}\ga^5V(2l+2)}{\al_{2l}\de U(2l+3)}+\dfrac{\ga^3(\ga-\de)}{\ep^{2l}\al_{2l}}\right)=0.
\end{align}
The equation above holds because of Equation \eqref{eqn:ODE_UV_Co}.
Therefore, Equation \eqref{eqn:g2l+2_f} holds.
\end{proof}

Let $\mathcal{K}$ be a field of rational functions over $\bbC$ defined as
\begin{equation}
 {\mathcal K}=\bbC(a_1,\dots,a_{2N},b,c,q^\frac{1}{2N-1}).
\end{equation}
Subsequently, from Lemmas \ref{lemma:f(l)_g(l)} and \ref{lemma:g_f}, the following lemma holds.

\begin{lemma}
Variables $f(l)$ and $g(l)$ are expressed as rational functions of variables $\{f_1,\dots,f_{2N}\}$ over ${\mathcal K}$.
\end{lemma}
\begin{proof}
From Equations \eqref{eqn:red_f_ODE} and \eqref{eqn:red_g_ODE}, the variables $f(l)$ and $g(l)$ are expressed as a rational function of $\{f_1,\dots,f_{2N}\}$ and that of $\{g_1,\dots,g_{2N}\}$ over ${\mathcal K}$, respectively.

The variables $g_i$, $i=1,\dots,2N$, are shown to be rational functions of $\{f_1,\dots,f_{2N}\}$ over ${\mathcal K}$, as described below.
From \eqref{eqn:g1_f}, it is evident that $g_1$ is a rational function of the variables $\{f_1,\dots,f_{2N}\}$.
From \eqref{eqn:g2l+3_f}, $g_{2i+3}$, $i=0,\dots,N-2$, are sequentially shown to be rational functions of the variables $\{f_1,\dots,f_{2N}\}$.
Finally, $g_{2j+2}$, $j=0,\dots,N-1$, are shown to be rational functions of variables $\{f_1,\dots,f_{2N}\}$ using \eqref{eqn:g2l+2_f}.
\end{proof}

We can easily verify that the conditions of reduction \eqref{eqn:(2N+1,1)_periodic_cond_1} and \eqref{eqn:(2N+1,1)_periodic_cond_3} are also invariant under the actions of the transformations $\{\si,w,\mu\}$.
Therefore, these transformations can be used even after reduction.
The actions of these transformations can be obtained from \eqref{eqns:def_sigma_w_mu}, \eqref{eqn:(2N+1,1)_periodic_cond_2}, \eqref{eqn:(2N+1,1)_periodic_cond_3}, \eqref{eqns:def_abcqfg}, and \eqref{eqn:def_fi_gj}.
These are given in the following lemma.

\begin{lemma}\label{lemma:birational_si_w_mu}
The actions of the transformations $\{\si,w,\mu\}$ on the parameters $\{a_0,\dots,a_{2N},b,c,q\}$ are given by
\begin{subequations}\label{eqns:sigma_w_mu_abcq}
\begin{align}
 &\si(a_i)=a_{i+2},\quad
 \si(b)=q^2b,\quad
 \si(c)=c,\quad
 \si(q)=q,\label{eqn:sigma_abcq}\\
 &w(a_i)=\dfrac{1}{a_{2N+1-i}},\quad
 w(b)=q^{2N+1}b,\quad
 w(c)=c^{-1},\quad
 w(q)=q^{-1},\label{eqn:sigma_w_mu_abcq}\\
 &\mu(a_i)=a_i,\quad
 \mu(b)=b,\quad
 \mu(c)=\dfrac{\ep}{c},\quad
 \mu(q)=q,\label{eqn:mu_abcq}
\end{align}
\end{subequations}
where $i\in\bbZ/(2N+1)\bbZ$ and $\ep=q^{\frac{2N+1}{1-2N}}$,
while their actions on the variables $\{f_1,\dots,f_{2N}\}$ are given by
\begin{subequations}\label{eqns:sigma_w_mu_f}
\begin{align}
 &\si(f_i)=f_{i+2},~i=1,\dots,2N-2,\quad
 \si(f_{2N-1})=f(2N+1),\quad
 \si(f_{2N})=f(2N+2),\label{eqn:sigma_f}\\
 &w(f_j)=f_{2N+1-j},\quad
 \mu(f_j)=g_j,\quad
 j=1,\dots,2N.\label{eqn:w_mu_f}
\end{align}
\end{subequations}
The variables $f(2N+1)$ and $f(2N+2)$ are given by \eqref{eqn:red_f_ODE}, and variables $\{g_1,\dots,g_{2N}\}$ are given by \eqref{eqns:g_f}.
\end{lemma}

\section{Proof of Theorem \ref{theorem:main}}\label{section:proof_mainThm}
In this section, we first review the transformation group $\widetilde{W}\left((A_{2N}\rtimes A_1)^{(1)}\right)$ given in \cite{nakazono2023higerA4A6}.
Then, using the transformations $\{\si,w,\mu\}$ given in \S \ref{section:transformations_si_w_mu}, we extend it to the transformation group $\widetilde{W}\left((A_{2N}\rtimes A_1)^{(1)}\times A_1^{(1)}\right)$.
\subsection{Review of the transformation group $\widetilde{W}\left((A_{2N}\rtimes A_1)^{(1)}\right)$}\label{subsection:A2NA1_review}
This subsection reviews the result in \cite{nakazono2023higerA4A6}.

Let us define the action of the transformations $\{s_1,\dots,s_{2N},w_1,\pi\}$ on the complex parameters $\{a_0,\dots,a_{2N},b,c,q\}$ and complex variables $\{f_1,\dots,f_{2N}\}$.
These parameters satisfy the following relation:
\begin{equation}\label{eqn:cond_ai_A2NA1}
 \prod_{i=0}^{2N}a_i=q.
\end{equation}
The actions of the transformations $\{s_1,\dots,s_{2N},w_1,\pi\}$ on the parameters $\{a_0,\dots,a_{2N},b,c,q\}$ are given by
\begin{subequations}\label{eqns:WA2N_para}
\begin{align}
 &s_i(a_j)
 =\begin{cases}
 {a_i}^{-1}&\text{if } j=i,\\
 a_ia_j&\text{if } j=i\pm 1,\\
 a_j&\text{otherwise},
 \end{cases}\qquad
 i=1,\dots,2N-1,\\
 &s_{2N}(a_j)
 =\begin{cases}
 {a_{2N}}^{-1}&\text{if } j=2N,\\
 a_{2N}a_j&\text{if } j=0,2N-1,\\
 a_j&\text{otherwise},
 \end{cases}\\
 &s_k(b)=b,\quad
 s_k(c)=c,\quad
 s_k(q)=q,\quad
 k=1,\dots,2N,\\
 &w_1(a_j)=\dfrac{1}{a_{2N+1-j}},\quad
 w_1(b)=q^{2N+1}b,\quad
 w_1(c)=c^{-1},\quad
 w_1(q)=q^{-1},\\
 &\pi(a_j)=a_{j+1},\quad
 \pi(b)=qb,\quad
 \pi(c)=c^{-1},\quad
 \pi(q)=q,
\end{align}
\end{subequations}
where $j\in\bbZ/(2N+1)\bbZ$,
while their actions on the variables $\{f_1,\dots,f_{2N}\}$ are given by
\begin{subequations}\label{eqns:WA2N_f}
\begin{align}
 &s_i(f_j)
 =\begin{cases}
 f_{i-1}\dfrac{\la(i-1)^2-{a_i}^{2N+1}f_i}{{a_i}^{2N+1}\la(i-1)^2-f_i}&\text{if } j=i-1,\\
 f_{i+1}\dfrac{{a_i}^{2N+1}\la(i-1)^2-f_i}{\la(i-1)^2-{a_i}^{2N+1}f_i}&\text{if } j=i+1,\\
 f_j&\text{otherwise},
 \end{cases}\qquad
 i,j=1,\dots,2N,\\
 &w_1(f_j)=f_{2N+1-j},\quad
 j=1,\dots,2N,\\
 &\pi(f_j)
 =\begin{cases}
 f_{j+1}&\text{if } j=1,\dots,2N-1,\\
 -\dfrac{c^4}{\displaystyle\prod_{k=1}^N f_{2k-1}}
 \left(
 \displaystyle\prod_{k=1}^N f_{2k}+\dfrac{q^{2N}a_0b}{\left(\displaystyle\prod_{k=1}^{2N-1}{a_k}^{2N-k}\right)c}
 \right)&\text{if } j=2N,
 \end{cases}
\end{align}
\end{subequations}
where
\begin{equation}
 \la(l)
 =\begin{cases}
 c&\text{(if $l$ is even)},\\
 c^{-1}&\text{(if $l$ is odd)}.
 \end{cases}
\end{equation} 
Define the transformation $s_0$ by
\begin{equation}
 s_0=\pi^{-1}s_1\pi.
\end{equation}
Then, the transformations $\{s_0,\dots,s_{2N},w_1,\pi\}$ satisfy the following relations:
\begin{subequations}\label{eqns:relations_si_w1_pi}
\begin{align}
 &{s_i}^2=(s_is_{i\pm 1})^3=(s_is_j)^2=1,~ j\neq i\pm 1,\quad
 {w_1}^2=1,\\
 &w_1 s_k=s_{2N-k+1}w_1,\quad
 \pi^\infty=1,\quad
 \pi s_k=s_{k+1}\pi,\quad
 \pi w_1=w_1\pi^{-1},
\end{align}
\end{subequations}
where $i,j,k\in\bbZ/(2N+1)\bbZ$.
Moreover, by defining the transformations $w_0$ and $r$ by
\begin{equation}\label{eqn:def_w0_r}
 w_0=\pi^2w_1,\quad
 r=\pi w_1,
\end{equation}
the transformations $\{s_0,\dots,s_{2N},w_0,w_1,r\}$ satisfy the following relations:
\begin{subequations}\label{eqns:fundamental_A2NA1}
\begin{align}
 &{s_i}^2=(s_is_{i\pm 1})^3=(s_is_j)^2=1,~ j\neq i\pm 1,\quad
 {w_0}^2={w_1}^2=(w_0w_1)^\infty=1,\label{eqn:fundamental_A2NA1_1}\\
 &w_0s_i=s_{2N-i+3}w_0,\quad
 w_1s_i=s_{2N-i+1}w_1,\label{eqn:fundamental_A2NA1_2}\\
 &r^2=1,\quad
 rs_i=s_{2N-i+2}r,\quad
 rw_0=w_1r,\quad
 rw_1=w_0r,\label{eqn:fundamental_A2NA1_3}
\end{align}
\end{subequations}
where $i,j\in\bbZ/(2N+1)\bbZ$.
From the relation \eqref{eqns:fundamental_A2NA1}, the following hold.
\begin{itemize}
\item 
Transformation groups $\langle s_0,\dots,s_{2N}\rangle$ and $\langle w_0,w_1\rangle$ form affine Weyl groups of type $A_{2N}^{(1)}$ and type $A_1^{(1)}$, respectively.
\item 
The transformation group $\langle s_0,\dots,s_{2N},w_0,w_1\rangle$ is a semidirect product of $\langle s_0,\dots,s_{2N}\rangle$ and $\langle w_0,w_1\rangle$, that is, 
\begin{equation}
 \langle s_0,\dots,s_{2N},w_0,w_1\rangle=\langle s_0,\dots,s_{2N}\rangle\rtimes\langle w_0,w_1\rangle.
\end{equation}
We denote it as $W\left((A_{2N}\rtimes A_1)^{(1)}\right)$.
\item 
The transformation $r$ corresponds to a reflection of the Dynkin diagram of type $A_{2N}^{(1)}$ associated with $\langle s_0,\dots,s_{2N}\rangle$ and to that of type $A_1^{(1)}$ associated with $\langle w_0,w_1\rangle$.
Therefore, we refer to the transformation group $\langle s_0,\dots,s_{2N},w_0,w_1,r\rangle$ as an extended affine Weyl group of type $(A_{2N}\rtimes A_1)^{(1)}$ and denote it as $\widetilde{W}\left((A_{2N}\rtimes A_1)^{(1)}\right)$, that is,
\begin{equation}
 \widetilde{W}\left((A_{2N}\rtimes A_1)^{(1)}\right)
 =\Big(\langle s_0,\dots,s_{2N}\rangle\rtimes\langle w_0,w_1\rangle\Big)\rtimes\langle r\rangle.
\end{equation}
\end{itemize}

\begin{remark}\label{remark:T_hT}
Define the transformations $T_0,T_1\in \widetilde{W}\left((A_{2N}\rtimes A_1)^{(1)}\right)$ as
\begin{equation}
 T_0=\pi^{-2N-1},\quad
 T_1=\pi s_{2N}s_{2N-1}\cdots s_1,
\end{equation}
whose actions on the parameters $\{a_0,\dots,a_{2N},b,c,q\}$ are given by
\begin{subequations}
\begin{align}
 &T_0(a_i)=a_i,~i=0,\dots,2N,\quad
 T_0(b)=q^{-2N-1}b,\quad
 T_0(c)=c^{-1},\quad
 T_0(q)=q,\\
 &T_1(a_j)
 =\begin{cases}
 qa_0&\text{if } j=0,\\
 q^{-1}a_1&\text{if } j=1,\\
 a_j&\text{otherwise},
 \end{cases}\quad
 T_1(b)=qb,\quad
 T_1(c)=c^{-1},\quad
 T_1(q)=q.
\end{align}
\end{subequations}
The system \eqref{eqn:intro_dP_odd} can be obtained from the action of $T_0$ with the following correspondence (see \cite{nakazono2023higerA4A6} for details):
\begin{equation}
 \overline{\rule{0em}{0.5em}~\,}=T_0,\quad
 F_i=f_i,\quad
 p=q^{-2N-1}.
\end{equation}
Thus, the transformation $T$ satisfying \eqref{eqn:intro_T} is given by
\begin{equation}
 T=T_0.
\end{equation}
Also, the transformation $\hat{T}$ satisfying \eqref{eqn:intro_hT} is given by
\begin{equation}
 \hat{T}={T_1}^{2N+1}T_0.
\end{equation}
Note that in \S \ref{Introduction}, we omit parameter $a_0$ using relation \eqref{eqn:cond_ai_A2NA1}, and we use $p=q^{-2N-1}$ instead of $q$.
\end{remark}

\begin{remark}\label{remark:pi_f}
The transformation $\pi\in\widetilde{W}\left((A_{2N}\rtimes A_1)^{(1)}\right)$ gives the following relation:
\begin{equation}
 \pi^{2N}(f_1)f_1\left(\prod_{k=1}^{N-1} \pi^{2k}(f_1)\right)
 +c^4\left(\prod_{k=1}^N \pi^{2k-1}(f_1)\right)
 +\dfrac{q^{2N}a_0b c^3}{\displaystyle\prod_{k=1}^{2N-1}{a_k}^{2N-k}}
 =0.
\end{equation}
Applying $\pi^l$ on the equation above and setting
\begin{equation}
 f(k)=\pi^{k-1}(f_1),
\end{equation}
we obtain the $2N$th-order $q$-{\ODE} \eqref{eqn:red_f_ODE}.
\end{remark}

\subsection{Proof of Theorem \ref{theorem:main}}\label{subsection:proof_mainThm}
In this subsection, using the transformations $\{\si,w,\mu\}$ given in \S \ref{section:transformations_si_w_mu} and transformation group $\widetilde{W}\left((A_{2N}\rtimes A_1)^{(1)}\right)$ shown in \S \ref{subsection:A2NA1_review}, we prove Theorem \ref{theorem:main}.
For this purpose, in this subsection, we discuss the proof under the following setup:
\begin{itemize}
\item 
The parameters $\{a_0,\dots,a_{2N},b,c,q\}$ and variables $\{f_1,\dots,f_{2N}\}$ defined in \S \ref{subsection:staircase_reduction} are considered identical to those in \S \ref{subsection:A2NA1_review}.
\item
We consider that the transformations $\{\si,w,\mu\}$ have actions defined only on the parameters $\{a_0,\dots,a_{2N},b,c,q\}$ and variables $\{f_1,\dots,f_{2N}\}$.
\end{itemize}

Comparing the actions of the transformations $\{\si,w\}$ \eqref{eqns:sigma_w_mu_abcq} and \eqref{eqns:sigma_w_mu_f} with those of $\{\pi,w_1\}$ \eqref{eqns:WA2N_para} and \eqref{eqns:WA2N_f}, we obtain: 
\begin{equation}\label{eqn:sigma=pi2_w=w1}
 \si=\pi^2,\quad
 w=w_1,
\end{equation}
which implies that $\si,w\in\widetilde{W}\left((A_{2N}\rtimes A_1)^{(1)}\right)$.
We note Remark \ref{remark:pi_f} for the action of $\pi$.
Next, we focus on transformation $\mu$.
It is evident from the actions on parameter $c$ that transformation $\mu$ is not contained within $\widetilde{W}\left((A_{2N}\rtimes A_1)^{(1)}\right)$. 
Therefore, in the following, we consider what type of extended affine Weyl group is formed by extending the transformation group $\widetilde{W}\left((A_{2N}\rtimes A_1)^{(1)}\right)$ through the addition of transformation $\mu$.

\begin{lemma}
The following relations hold:
\begin{equation}\label{eqn:relations_si_w1_pi_mu}
 s_1 \mu=\mu s_1,\quad
 w_1\mu=\mu w_1,\quad
 \pi^2\mu=\mu\pi^2,\quad
 (\mu \pi)^\infty=1.
\end{equation}
\end{lemma}
\begin{proof}
Because of \eqref{eqn:sigma=pi2_w=w1} and \eqref{eqn:relations_sigma_w_mu},
the relations $w_1\mu=\mu w_1$ and $\pi^2\mu=\mu\pi^2$ hold.
For a positive integer $k$, the following holds:
\begin{equation}
 (\mu \pi)^k(c)=\ep^{-k}c,
\end{equation}
which gives $(\mu \pi)^\infty=1$.

Let us prove that the relation $s_1 \mu=\mu s_1$ holds.
Because the relation obviously holds for the action on the parameters, we consider the action on the $f$-variables.
Applying $s_1$ to Equation \eqref{eqn:g1_f}, we obtain
\begin{align}
 &\left(1-\dfrac{{a_1}^{2N+1}f_1}{c^{2}}\right)
 \dfrac{f_1-\ep^{2N-3}c^{6-4N}s_1(g_1)}{f_1-q^{2N+1}\ep^{2N-3}c^4s_1(g_1)}\notag\\
 &+\dfrac{{a_1}^{2N+1}(c^2-{a_1}^{2N+1}f_1)}{{a_1}^{2N+1}c^2-f_1}\left\{\,
 \sum_{k=1}^{N-1}\dfrac{\displaystyle\prod_{i=1}^k f_{2i-1}}{\left(\displaystyle\prod_{i=1}^{2k}{a_i}^{2N+1}\right)c^{4k} \left(\displaystyle\prod_{i=1}^k f_{2i}\right)}\left(1-\dfrac{f_{2k+1}}{{a_{2k+1}}^{2N+1}c^2}\right)\right.\notag\\
 &\left.\hspace{10em}
 +\dfrac{{a_0}^{2N+1}}{q^{2N+1}c^{4N}\left(\displaystyle\prod_{i=1}^N f_{2i}\right)}\left(\rule{0em}{2.5em}\right.\dfrac{q^{2N}bc}{\displaystyle\prod_{i=0}^{2N-1}{a_i}^{2N-i}}+\displaystyle\prod_{i=1}^N f_{2i-1}\left.\rule{0em}{2.5em}\right)\,\right\}
 =0.
\end{align}
By eliminating the second and third terms of Equation \eqref{eqn:g1_f} using the equation above, we obtain
\begin{align}
 &\left(1-\dfrac{f_1}{{a_1}^{2N+1}c^{2}}\right)
 \dfrac{f_1-\ep^{2N-3}c^{6-4N}g_1}{f_1-q^{2N+1}\ep^{2N-3}c^4g_1}\notag\\
 &=\dfrac{{a_1}^{2N+1}c^2-f_1}{{a_1}^{2N+1}(c^2-{a_1}^{2N+1}f_1)}\,
 \left(1-\dfrac{{a_1}^{2N+1}f_1}{c^{2}}\right)
 \dfrac{f_1-\ep^{2N-3}c^{6-4N}s_1(g_1)}{f_1-q^{2N+1}\ep^{2N-3}c^4s_1(g_1)},
\end{align}
which, upon simplification, gives 
\begin{equation}
 s_1(g_1)=g_1.
\end{equation}
Furthermore, by comparing the equation obtained by applying $s_1$ to Equation \eqref{eqn:g2l+3_f}$_{l_1=0}$ with Equation \eqref{eqn:g2l+3_f}$_{l_1=0}$, we obtain
\begin{equation}
 s_1(g_3)=g_3.
\end{equation}
Similarly, we can prove
\begin{equation}
 s_1(g_{2l_1+3})=g_{2l_1+3},\quad l_1=1,\dots,N-2,
\end{equation}
sequentially from Equation \eqref{eqn:g2l+3_f}.
Finally, by comparing the equation obtained by applying $s_1$ to Equation \eqref{eqn:g2l+2_f} with Equation \eqref{eqn:g2l+2_f}, we obtain:
\begin{equation}
 s_1(g_{2l_2+2})=g_{2l_2+2},\quad l_2=0,\dots,N-1.
\end{equation}
Therefore, we obtain
\begin{equation}
 s_1 \mu(f_l)=s_1(g_l)=g_l=\mu(f_l)=\mu s_1(f_l),\quad 
 l=1,\dots,2N.
\end{equation}
Thus, $s_1 \mu=\mu s_1$ holds.
\end{proof}

Defining the transformation $\mu_0$ and $\mu_1$ as
\begin{equation}\label{eqn:def_mu0mu1}
\mu_0=\pi^{-1}\mu\pi,\quad
 \mu_1=\mu,
\end{equation}
the following lemma holds.

\begin{lemma}\label{lemma:fundamental_mu}
The following relations hold:
\begin{equation}\label{eqn:fundamental_mu}
 {\mu_i}^2=(\mu_0\mu_1)^\infty=1,\quad
 \mu_i s_j=s_j \mu_i,\quad
 \mu_i w_k=w_k \mu_i,\quad
 r\mu_0=\mu_1 r,\quad
 r\mu_1=\mu_0 r,
\end{equation}
where 
\begin{equation}
 i=0,1,\quad
 j=0,\dots,2N,\quad
 k=0,1.
\end{equation}
\end{lemma}
\begin{proof}
Relation $(\mu_0\mu_1)^\infty=1$ holds becuse for positive integer $k$, the following holds:
\begin{equation}
 (\mu_0\mu_1)^k(c)=\ep^{2k}c.
\end{equation}
The other relations are shown using Equations \eqref{eqn:relations_sigma_w_mu}, \eqref{eqns:relations_si_w1_pi}, \eqref{eqn:def_w0_r}, \eqref{eqns:fundamental_A2NA1}, and \eqref{eqn:relations_si_w1_pi_mu} as follows: 
\begin{subequations}
\begin{align}
 &{\mu_0}^2
 =(\pi^{-1}\mu\pi)(\pi^{-1}\mu\pi)
 =\pi^{-1}\mu^2\pi
 =1,\\
 &{\mu_1}^2=\mu^2=1,\\
 &r\mu_0
 =(\pi w_1)(\pi^{-1}\mu\pi)
=w_1 \pi^{-2}\mu\pi
=w_1 \mu\pi^{-1}
=\mu w_1 \pi^{-1}
=\mu (\pi w_1)
=\mu_1 r,\\
 &r\mu_1
 =r\mu_1r^2
 =r (\mu_1 r) r
 =r (r \mu_ 0) r
 =r^2 \mu_ 0 r
 =\mu_ 0 r,\\
 &\mu_1 w_1
 =\mu w_1
 =w_1 \mu
 =w_1 \mu_1,\\
 &\mu_1 w_0
 =\mu (\pi^2 w_1)
 =(\pi^2 w_1)\mu
 =w_0 \mu_1,\\
 &\mu_0 w_1
 =\mu_0 r^2 w_1
 =(\mu_0 r) (r w_1)
 =r(\mu_1 w_0)r
 =r(w_0 \mu_1)r
 =(w_1 r)(r \mu_0)
 =w_1 \mu_0,\\
 &\mu_0 w_0
 =\mu_0 r^2 w_0
 =(\mu_0 r)(r w_0)
 =r(\mu_1 w_1)r
 =r(w_1\mu_1)r
 =(w_0 r)(r \mu_0)
 =w_0\mu_0,\\
 &\mu_1 s_1
 =\mu s_1
 =s_1 \mu
 =s_1\mu_1,\\
 &\mu_1 s_{2i+1}
 =\mu (\pi^{2i}s_1\pi^{-2i})
 =(\pi^{2i}s_1\pi^{-2i})\mu
 =s_{2i+1}\mu_1,\\
 &\mu_1 s_{2N}
 =\mu_1 {w_1}^2 s_{2N}
 =(\mu_1 w_1)(w_1 s_{2N})
 =w_1 (\mu_1 s_1) w_1
 =w_1 (s_1 \mu_1) w_1
 =s_{2N}\mu_1,\\
 &\mu_1 s_{2j}
 =\mu_1(\pi^{2(j-N)}s_{2N} \pi^{2(N-j)})
 =(\pi^{2(j-N)}s_{2N} \pi^{2(N-j)})\mu_1
 =s_{2j}\mu_1,\\
 &\mu_0 s_k
 =(\pi^{-1}\mu_1 \pi) s_k
 =\pi^{-1}\mu_1 s_{k+1} \pi
 =\pi^{-1} s_{k+1} \mu_1 \pi 
 =s_k(\pi^{-1}\mu_1 \pi)
 =s_k \mu_0,
\end{align}
\end{subequations}
where
\begin{equation}
 i=1,\dots,N-1,\quad
 j=0,\dots,N-1,\quad
 k\in\bbZ/(2N+1)\bbZ.
\end{equation}
\end{proof}

From equation \eqref{eqn:fundamental_mu}, we find that the transformation group $\langle \mu_0,\mu_1\rangle$ forms the affine Weyl group of type $A_1^{(1)}$, and is orthogonal to the transformation group
\begin{equation}
 W\left((A_{2N}\rtimes A_1)^{(1)}\right)=\langle s_0,\dots,s_{2N}\rangle\rtimes\langle w_0,w_1\rangle.
\end{equation}
The transformation $r$ corresponds not only to a reflection of the Dynkin diagram of type $A_{2N}^{(1)}$ associated with $\langle s_0,\dots,s_{2N}\rangle$ and that of type $A_1^{(1)}$ associated with $\langle w_0,w_1\rangle$, but also to a reflection of the Dynkin diagram of type $A_1^{(1)}$ associated with $\langle \mu_0,\mu_1\rangle$.
Therefore, we refer to the transformation group
\begin{equation}
 \langle s_0,\dots,s_{2N},w_0,w_1,\mu_0,\mu_1,r\rangle=\left(W\left((A_{2N}\rtimes A_1)^{(1)}\right)\times \langle \mu_0,\mu_1\rangle\right)\rtimes\langle r\rangle
\end{equation}
as an extended affine Weyl group of type $(A_{2N}\rtimes A_1)^{(1)}\times A_1^{(1)}$ and denote it as 
\begin{equation}
 \widetilde{W}\left((A_{2N}\rtimes A_1)^{(1)}\times A_1^{(1)}\right),
\end{equation}
that is,
\begin{equation}
 \widetilde{W}\left((A_{2N}\rtimes A_1)^{(1)}\times A_1^{(1)}\right)
 =\bigg(\Big(\langle s_0,\dots,s_{2N}\rangle\rtimes\langle w_0,w_1\rangle\Big)\times\langle \mu_0,\mu_1\rangle\bigg)\rtimes\langle r\rangle.
\end{equation}
Thus, we have completed the proof of Theorem \ref{theorem:main}.

\begin{remark}\label{remark:A4_A2A1A1}
In the case $N=1$, the group $\widetilde{W}\left((A_2\rtimes A_1)^{(1)}\times A_1^{(1)}\right)$ can be reconstructed as a subgroup of the extended affine Weyl symmetry group of type $A_4^{(1)}$ for $A_4^{(1)}$-surface type $q$-Painlev\'e equations \cite{SakaiH2001:MR1882403,TsudaT2006:MR2247459,JNS2016:MR3584386}.
Indeed, all generators of $\widetilde{W}\left((A_2\rtimes A_1)^{(1)}\times A_1^{(1)}\right)$ are given by the elements of the transformation group in \cite{JNS2016:MR3584386}:
\begin{equation}
 \widetilde{W}\left(A_4^{(1)}\right)=\langle {\bm s_0},{\bm s_1},{\bm s_2},{\bm s_3},{\bm s_4}\rangle \rtimes \langle{\bm \sigma},{\bm \iota}\rangle
\end{equation}
as the following:
\begin{subequations}
\begin{align}
 &s_0={\bm s_0},\quad
 s_1={\bm s_3}{\bm s_4}{\bm s_3},\quad
 s_2={\bm s_1}{\bm s_2}{\bm s_1},\quad
 w_0={\bm \sigma}{\bm \iota}{\bm s_2}{\bm s_4},\quad
 w_1={\bm \iota},\\
 &\mu_0={\bm s_4}{\bm s_0}{\bm s_1}{\bm s_0}{\bm s_4},\quad
 \mu_1={\bm s_2}{\bm s_3}{\bm s_2},\quad
 r={\bm \sigma^3}{\bm \iota}{\bm s_4}. 
\end{align}
\end{subequations}
In this case, the correspondence between the parameters $\{a_0,a_1,a_2,b,c,q\}$ and variables $\{f_1,f_2\}$ on which $\widetilde{W}\left((A_2\rtimes A_1)^{(1)}\times A_1^{(1)}\right)$ acts and the parameters $\{{\bm a_0},{\bm a_1},{\bm a_2},{\bm a_3},{\bm a_4},{\bm q}\}$ and variables $\{{\bm f}_1^{(i)},{\bm f}_2^{(i)}\}_{i=1,\dots,5}$ on which $\widetilde{W}\left(A_4^{(1)}\right)$ acts is given below.
\begin{subequations}
\begin{align}
 &a_0={\bm a_0}^{-1/6},\quad
 a_1={\bm a_3}^{-1/6}{\bm a_4}^{-1/6},\quad
 a_2={\bm a_1}^{-1/6}{\bm a_2}^{-1/6},\\
 &b=-\dfrac{{\bm a_0}^{1/4}{\bm a_1}^{7/12}{\bm a_2}^{1/12}{\bm a_3}^{5/12}}{{\bm a_4}^{1/12}},\quad
 c=\dfrac{{\bm a_0}^{1/4}{\bm a_1}^{1/4}{\bm a_4}^{1/4}}{{\bm a_2}^{1/4}{\bm a_3}^{1/4}},\quad
 q={\bm q}^{-1/6},\\
 &f_1=\dfrac{{\bm a_0}^{1/2}{\bm a_1}^{1/2}}{{\bm a_2}^{1/2}}\left(1-\dfrac{{\bm a_0}{\bm a_4}}{{\bm a_2}{\bm a_3}}{\bm f}_1^{(2)}{\bm f}_1^{(3)}\right),\quad
 f_2=-\dfrac{{\bm a_0}^{1/2}}{{\bm a_3}^{1/2}{\bm a_4}^{1/2}}{\bm f}_1^{(3)}.
\end{align}
\end{subequations}
\end{remark}

\begin{remark}\label{remark:wT}
The transformation $\widetilde{T}$ in \S \ref{Introduction} is given by, for example, 
\begin{equation}
 \widetilde{T}=\mu_0\mu_1\in\widetilde{W}\left((A_{2N}\rtimes A_1)^{(1)}\times A_1^{(1)}\right)
\end{equation}
whose action on the parameters $\{a_0,\dots,a_{2N},b,c,q\}$ is
\begin{equation}
 \widetilde{T}(a_i)=a_i,~i=0,\dots,2N,\quad
 \widetilde{T}(b)=b,\quad
 \widetilde{T}(c)=\ep^2c,\quad
 \widetilde{T}(q)=q,
\end{equation}
where $\ep=q^{\frac{2N+1}{1-2N}}$.
\end{remark}

\section{Concluding remarks}\label{ConcludingRemarks}
In this study, we extended the birational representation of an extended affine Weyl group of $(A_{2N}\rtimes A_1)^{(1)}$-type, the symmetry of the system \eqref{eqn:intro_dP_odd} obtained in \cite{nakazono2023higerA4A6}, to a birational representation of an extended affine Weyl group of $(A_{2N}\rtimes A_1)^{(1)}\times A_1^{(1)}$-type.
For this purpose, we used the system \eqref{eqns:PDEs_UV} with the CAC property.

Surprisingly, a Painlev\'e-type {\ODE} is related to a system of {\PDE} such as system \eqref{eqns:PDEs_UV}, which consists of an alternation of two types of CAC-cubes (see Appendix \ref{section:system_CAC_proof}).
This finding suggests that the theory of CAC property is deeply connected to the theory of Painlev\'e type {\ODE}s.
We provide conjectures on periodic reductions from systems of {\PDE}s with the CAC property to Painlev\'e type {\ODE}s in Appendix \ref{section:conjecture}.

\subsection*{Acknowledgment}
This work was supported by JSPS KAKENHI, Grant Numbers JP19K14559 and JP23K03145.
\subsection*{Data availability statement}
Any data that support the findings of this study are included within the article.
\appendix
\section{CAC and tetrahedron properties of the system \eqref{eqns:PDEs_UV}}\label{section:system_CAC_proof}
This Appendix shows that system \eqref{eqns:PDEs_UV} has the CAC and tetrahedron properties.
(See \cite{BS2002:MR1890049,NijhoffFW2002:MR1912127,WalkerAJ:thesis,NW2001:MR1869690} for the CAC and tetrahedron properties.)

Let us consider system \eqref{eqns:PDEs_UV} as a system of equations around cubes, as follows:
Assign the variables $U_{l,m}$, $V_{l,m}$ to vertices $(l,m,0),(l,m,1)\in\bbZ^3$, respectively.
Subsequently, for fixed $l,m\in\bbZ$, the system \eqref{eqns:PDEs_UV} can be regarded as equations on the faces (face equations) of the cube, whose vertices are given by:
\begin{align}\label{eqn:lattice_lm0}
 &(l,m,0),\quad
 (l+1,m,0),\quad
 (l,m+1,0),\quad
 (l+1,m+1,0),\notag\\
 &(l,m,1),\quad
 (l+1,m,1),\quad
 (l,m+1,1),\quad
 (l+1,m+1,1).
\end{align}
Here, we label the six faces of the cube as ${\mathcal A}$, ${\mathcal A}'$, ${\mathcal B}$, ${\mathcal B}'$, ${\mathcal C}$, ${\mathcal C}'$ as illustrated in Figure \ref{fig:cubeUV}.
Subsequently, from the system \eqref{eqns:PDEs_UV} the six face equations are given by
\begin{subequations}\label{eqns:UV_cube_even}
\begin{align}
 &{\mathcal A}:~\dfrac{~\widetilde{\overline{U}}~}{U}+\dfrac{\ga^4\widetilde{U}}{\overline{U}}+\al_l\,\be_m\ga=0,
 &&{\mathcal A}':~\dfrac{~\widetilde{\overline{V}}~}{V}+\dfrac{\de^4\widetilde{V}}{\overline{V}}+\al_l\,\be_m\de=0,\\
 &{\mathcal B}:~\dfrac{~\widetilde{U}~}{V}=\dfrac{~\widetilde{V}~}{U},
 &&{\mathcal B}':~\dfrac{~\widetilde{\overline{U}}~}{\overline{V}}=\dfrac{\ga^4}{\de^4}\left(\dfrac{~\widetilde{\overline{V}}~}{\overline{U}}+\dfrac{\be_m\de^4(\ga-\de)}{\ep^{l-2m}}\right),\\
 &{\mathcal C}:~\dfrac{~U~}{\overline{V}}=\dfrac{\ga^3}{\de^3}\left(\dfrac{~V~}{\overline{U}}-\dfrac{\de^3(\ga-\de)}{\al_l\ep^{l-2m}}\right),
 &&{\mathcal C}':~\dfrac{~\widetilde{U}~}{\widetilde{\overline{V}}}=\dfrac{\ga}{\de}\left(\dfrac{~\widetilde{\overline{U}}~}{\widetilde{V}}+\dfrac{\ga(\ga-\de)}{\al_l\ep^{l-2m-3}}\right)^{-1},
\end{align}
\end{subequations}
if $(l+m)$ is even,
and by
\begin{subequations}\label{eqns:UV_cube_odd}
\begin{align}
 &{\mathcal A}:~\dfrac{~\widetilde{\overline{U}}~}{U}+\dfrac{\widetilde{U}}{\ga^4\overline{U}}+\dfrac{\al_l\,\be_m}{\ga}=0,
 &&{\mathcal A}':~\dfrac{~\widetilde{\overline{V}}~}{V}+\dfrac{\widetilde{V}}{\de^4\overline{V}}+\dfrac{\al_l\,\be_m}{\de}=0,\\
 &{\mathcal B}:~\dfrac{~\widetilde{U}~}{V}=\dfrac{\ga^4}{\de^4}\left(\dfrac{~\widetilde{V}~}{U}+\dfrac{\be_m\de^4(\ga-\de)}{\ep^{l-2m-1}}\right),
 &&{\mathcal B}':~\dfrac{~\widetilde{\overline{U}}~}{\overline{V}}=\dfrac{~\widetilde{\overline{V}}~}{\overline{U}},\\
 &{\mathcal C}:~\dfrac{~U~}{\overline{V}}=\dfrac{\ga}{\de}\left(\dfrac{~\overline{U}~}{V}+\dfrac{\ga(\ga-\de)}{\al_l\ep^{l-2m-1}}\right)^{-1},
 &&{\mathcal C}':~\dfrac{~\widetilde{U}~}{\widetilde{\overline{V}}}=\dfrac{\ga^3}{\de^3}\left(\dfrac{~\widetilde{V}~}{\widetilde{\overline{U}}}-\dfrac{\de^3(\ga-\de)}{\al_l\ep^{l-2m-2}}\right),
\end{align}
\end{subequations}
if $(l+m)$ is odd.
Here,
\begin{subequations}
\begin{align}
 &U=U_{l,m},\quad
 \overline{U}=U_{l+1,m},\quad
 \widetilde{U}=U_{l,m+1},\quad
 \widetilde{\overline{U}}=U_{l+1,m+1},\\
 &V=V_{l,m},\quad
 \overline{V}=V_{l+1,m},\quad
 \widetilde{V}=V_{l,m+1},\quad
 \widetilde{\overline{V}}=V_{l+1,m+1}.
\end{align}
\end{subequations}
Denote the cubes given by \eqref{eqns:UV_cube_even} and \eqref{eqns:UV_cube_odd} as $C^{\rm (e)}_{l,m}$ and $C^{\rm (o)}_{l,m}$, respectively. 
By direct calculation, we can verify that both cubes $C^{\rm (e)}_{l,m}$ and $C^{\rm (o)}_{l,m}$ have the CAC and tetrahedron properties.
Because all cubes composing the system \eqref{eqns:PDEs_UV} have the CAC and tetrahedron properties, system \eqref{eqns:PDEs_UV} is said to have the CAC and tetrahedron properties.
Note that the tetrahedron relations of the cube $C^{\rm (e)}_{l,m}$ are given by
\begin{subequations}
\begin{align}
 &\overline{U}\,\widetilde{U}-\dfrac{\al_l\ep^{l-2m-3}}{\de(\ga-\de)}\left(\overline{U}\,\widetilde{\overline{V}}+\ga^2\ep\,\widetilde{U}\,V+\al_l\,\be_m\,\de\,\overline{U}\,V\right)=0,\\
 &\overline{V}\,\widetilde{V}+\dfrac{\al_l\ep^{l-2m-3}}{\ga(\ga-\de)}\left(\widetilde{\overline{U}}\,\overline{V}+\de^2\ep\,U\,\widetilde{V}+\al_l\,\be_m\,\ga\,U\,\overline{V}\right)=0,
\end{align}
\end{subequations}
while those of the cube $C^{\rm (o)}_{l,m}$ are given by
\begin{subequations}
\begin{align}
 &V\,\widetilde{\overline{V}}+\dfrac{\al_l\ep^{l-2m-2}}{\ga^3(\ga-\de)}\left(\widetilde{U}\,V+\ga^2\ep\,\overline{U}\,\widetilde{\overline{V}}+\al_l\,\be_m\,\ga^3\overline{U}\,V\right)=0,\\
 &U\,\widetilde{\overline{U}}-\dfrac{\al_l\ep^{l-2m-2}}{\de^3(\ga-\de)}\left(U\,\widetilde{V}+\de^2\ep\,\widetilde{\overline{U}}\,\overline{V}+\al_l\,\be_m\,\de^3U\,\overline{V}\right)=0.
\end{align}
\end{subequations}

\begin{figure}[htbp]
\begin{center}
 \includegraphics[width=0.5\textwidth]{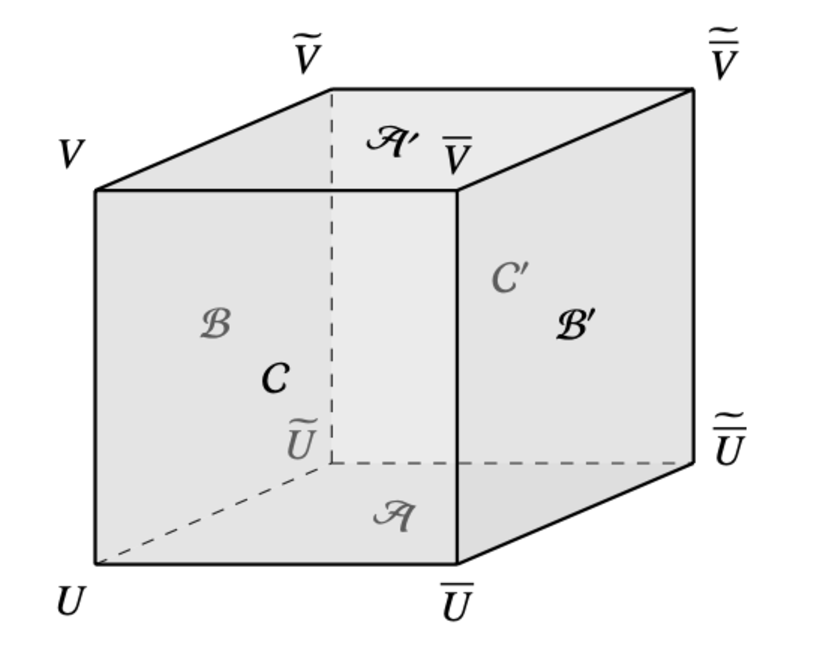}
\end{center}
\caption{A cube with 8 variables labelled by $\{U,\overline{U},\widetilde{U},\widetilde{\overline{U}},V,\overline{V},\widetilde{V},\widetilde{\overline{V}}\}$ and 6 faces labelled by $\{{\mathcal A},{\mathcal A}',{\mathcal B},{\mathcal B}',{\mathcal C},{\mathcal C}'\}$.
Bottom: ${\mathcal A}$,
top: ${\mathcal A}'$,
left: ${\mathcal B}$,
right: ${\mathcal B}'$,
front: ${\mathcal C}$,
back: ${\mathcal C}'$.
}
\label{fig:cubeUV}
\end{figure}

\begin{remark}
Adler-Bobenko-Suris \cite{ABS2003:MR1962121,ABS2009:MR2503862} and Boll \cite{BollR2012:MR3010833,BollR2011:MR2846098,BollR:thesis} classified six face equations of cubes with the CAC and tetrahedron properties.
The face equations given by systems \eqref{eqns:UV_cube_even} and \eqref{eqns:UV_cube_odd} belong to $D_4$ of the type $H^6$ in Boll's classification.
\end{remark}

\begin{remark}
Conversely, system \eqref{eqns:PDEs_UV} can be regarded as a system of {\PDE}s obtained by space-filling two types of cubes $C^{\rm (e)}_{l,m}$ and $C^{\rm (o)}_{l,m}$ into the sublattice of lattice $\bbZ^3$ whose vertices are given by \eqref{eqn:lattice_lm0}.
The cubes $C^{\rm (e)}_{l,m}$ and $C^{\rm (o)}_{l,m}$ are arranged in quadrilaterals of lattice $\bbZ^2$ as illustrated in Figure \ref{fig:2cubes}.
\end{remark}

\begin{figure}[htbp]
\begin{center}
 \includegraphics[width=0.55\textwidth]{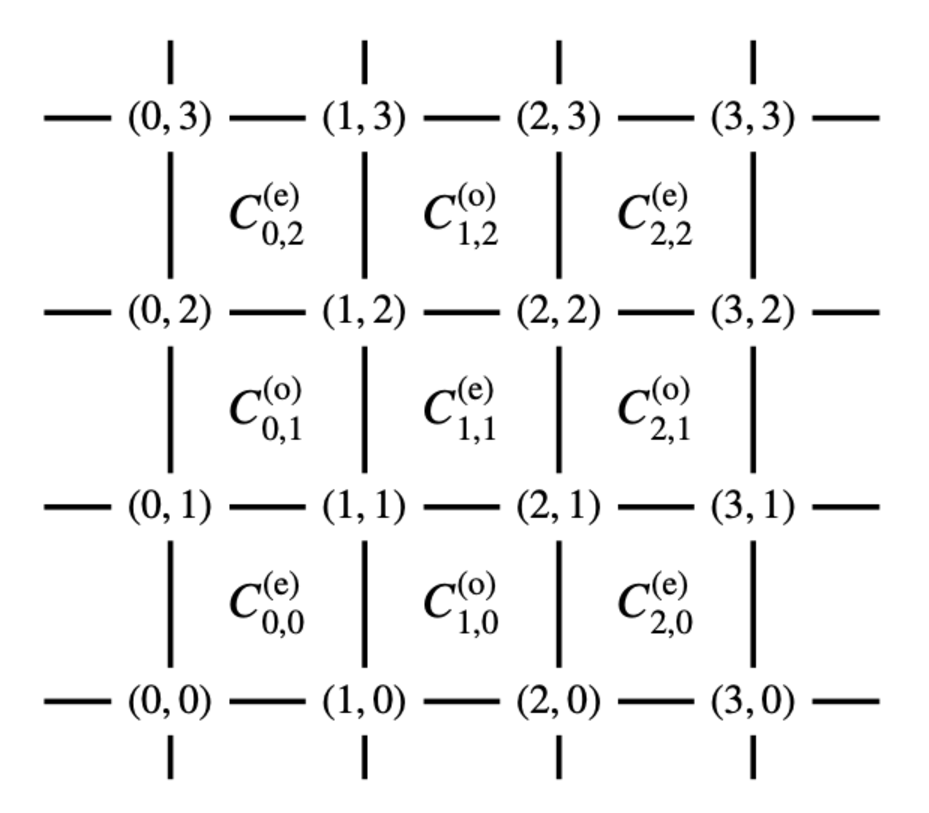}
\end{center}
\caption{Arrangement of the cubes $C^{\rm (e)}_{l,m}$ and $C^{\rm (o)}_{l,m}$.
These cubes are alternately positioned on quadrilaterals of the lattice $\bbZ^2$.}
\label{fig:2cubes}
\end{figure}

\section{Conjectures on periodic reductions from systems of {\PDE}s with the CAC property to Painlev\'e type $q$-{\ODE}s}\label{section:conjecture}
Fix the integers $n\geq 1$ and $L\geq 0$.
In this appendix, we label a pair, 
the system of {\PDE}s
\begin{subequations}\label{eqns:ABS_U}
\begin{align}
 &\la(l_{0\cdots n})^2\dfrac{U_{\overline{ij}}}{U}=\dfrac{\al^{(i)}(l_i) U_{\overline{i}}-\al^{(j)}(l_j) U_{\overline{j}}}{\al^{(j)}(l_j) U_{\overline{i}}-\al^{(i)}(l_i) U_{\overline{j}}},
 &&i<j,\quad i,j\in\{1,\dots,n\},\label{eqn:ABS_U_ij}\\
 &\dfrac{U_{\overline{0\,k}}}{U}+\la(l_{0\cdots n})^4\dfrac{U_{\overline{0}}}{U_{\overline{k}}}+\al^{(k)}(l_k) \ka(l_0) \la(l_{0\cdots n})=0,
 &&k=1,\dots,n,\label{eqn:ABS_U_0k}
\end{align}
\end{subequations}
and the periodic condition
\begin{equation}\label{eqn:ABS_U_cond}
 U(l_1+1,\dots,l_n+1,l_0+L)=U(l_1,\dots,l_n,l_0),
\end{equation} 
as an $(n,L)$-pair.
Note that the $(n,0)$- and $(1,L)$-pairs are exceptions and are defined as follows:
The $(n,0)$-pair is defined by a pair of the system \eqref{eqns:ABS_U} and the condition 
\begin{subequations}\label{eqn:ABS_U_cond_L=0}
\begin{align}
 &U(l_1,\dots,l_n,l_0)=G(l_1,\dots,l_n,l_0)\omega(l_1,\dots,l_n,l_0),\\
 &\omega(l_1+1,\dots,l_n+1,l_0)=\omega(l_1,\dots,l_n,l_0),
\end{align}
\end{subequations}
where, depending on the value of $n$, it is necessary to specify an appropriate gauge function $G(l_1,\dots,l_n,l_0)$.
(See \cite{JN2016:MR3597921} for the $(3,0)$-pair.)
On the other hand, the $(1,L)$-pair is defined as follow:
\begin{align}
 &\dfrac{U(l_1+1,l_0+1)}{U(l_1,l_0)}+\la(l_0+l_1)^4\dfrac{U(l_1,l_0+1)}{U(l_1+1,l_0)}+\al^{(1)}(l_1) \ka(l_0) \la(l_0+l_1)=0,\label{eqn:D4}\\
 &U(l_1+1,l_0+L)=U(l_1,l_0).\label{eqn:D4_cond}
\end{align}
Here, $l_0,\dots,l_n\in\bbZ$ are lattice parameters, 
$\{\al^{(1)}(l),\dots,\al^{(n)}(l),\ka(l),\la_0\}_{l\in\bbZ}$ are complex parameters and
\begin{subequations}
\begin{align}
 &U=U(l_1,\dots,l_n,l_0),\quad
 U_{\overline{i}}=U|_{\,l_i\to l_i+1},\quad
 U_{\overline{ij}}=U|_{(l_i,l_j)\to (l_i+1,l_j+1)},\\
 &l_{0\cdots n}:=\sum_{i=0}^nl_i,\quad
 \la(l)=\begin{cases}
 \la_0&l\in2\bbZ,\\[0.5em]
 \dfrac{1}{\la_0}&\text{otherwise}.
 \end{cases} 
\end{align}
\end{subequations}

\begin{remark}
The system \eqref{eqns:ABS_U} has the CAC and tetrahedron properties.
Equation \eqref{eqn:ABS_U_ij} with fixed $i$ and $j$ is known as the lattice modified KdV equation \cite{NC1995:MR1329559,NQC1983:MR719638} and belongs to H3 in ABS's classification \cite{ABS2003:MR1962121,ABS2009:MR2503862}, whereas Equation \eqref{eqn:ABS_U_0k} with fixed $k$ belongs to $D_4$ of the type $H^6$ in Boll's classification \cite{BollR2011:MR2846098,BollR2012:MR3010833,BollR:thesis}.
\end{remark}

\subsection{Conjectures on the $(n,L)$-pair}
This subsection presents conjectures on the $(n,L)$-pair.

Before providing conjectures on the $(n,L)$-pair, let us explain its results.
\begin{itemize}
\item 
In \cite{nakazono2023higerA4A6}, we obtained a birational representation of an extended affine Weyl group of $(A_{n-1}\rtimes A_1)^{(1)}$-type from the $(n,1)$-pair, where $n\geq 2$.
Note that the $(2,1)$-pair and $(3,1)$-pair were previously studied in \cite{JNS2015:MR3403054} and \cite{JNS2016:MR3584386}, respectively.
In particular, the birational representation for the $(2N+1,1)$-pair, where $N\geq 1$, was extended in this study to a birational representation of an extended affine Weyl group of type $(A_{2N}\rtimes A_1)^{(1)}\times A_1^{(1)}$. 
\item 
From the $(1,L)$-pair, where $L\geq 2$, we obtain the Painlev\'e type $q$-{\ODE}s given in \cite{Okubo2017:arxiv1704.05403,MOT2021:Cluster}.
See \S \ref{subequation:(1,L)} for details.
\item
We can derive $q$-Painlev\'e equations of $A_5^{(1)}$-surface type from the $(3,0)$-pair (see \cite{JN2016:MR3597921,JNS2015:MR3403054}).
Moreover, we can obtain $q$-Painlev\'e equations of $A_6^{(1)}$- and $A_4^{(1)}$-surface types from the $(2,1)$- and $(3,1)$-pairs, respectively (see \cite{JNS2016:MR3584386}) and 
$q$-Painlev\'e equations of $A_7^{(1)}$- and $A_6^{(1)}$-surface types from the $(1,2)$- and $(1,3)$-pairs, respectively (see \S \ref{subequation:(1,L)}).
\end{itemize}
The following are conjectures.
\begin{itemize}
\item 
We expect that the $(A_{n-1}\times A_1)^{(1)}$-type KNY's representation can be derived from the $(n,0)$-pair, where $n\geq 3$, in the same manner as in \cite{nakazono2023higerA4A6}.
(It has been shown in \cite{JN2016:MR3597921,JNS2015:MR3403054,JNS2014:MR3291391} that this conjecture is correct for the $(3,0)$-pair.)
In particular, the birational representation for the $(2N+2,0)$-pair, where $N\geq 1$, can be extended to the $(A_{2N+1}\times A_1\times A_1)^{(1)}$-type extended KNY's representation by using a different pair of a system of {\PDE}s and a periodic condition from the $(2N+2,0)$-pair, as demonstrated in this study.
\item 
We expect to obtain $q$-Painlev\'e equations of $A_5^{(1)}$- and $A_3^{(1)}$-surface types from the $(2,2)$- and $(4,0)$-pairs, respectively.
\item 
Through a limit operation, the $(n, L)$-pair can be degenerated to the $(n-1, L+1)$-pair. 
As illustrations, limit operations from the $(3, L)$-pair, where $L\neq0$, to the $(2, L+1)$-pair and from the $(2, L)$-pair, where $L\neq0$, to the $(1, L+1)$-pair are demonstrated in \S \ref{subsection:(3,L)(2,L)_reductions}.
\end{itemize}
The results and conjectures mentioned above and the predictions for the other pairs are summarized in Table \ref{table:nL}.
Note that everything not mentioned in the above results is a conjecture.
The following items briefly explain Table \ref{table:nL}.
\begin{itemize}
\item
$q$-{\ODE}s appearing at the positions marked with ``$\times$" in the table are linearizable. 
Moreover, there exist $q$-Painlev\'e equations of $A_i^{(1)}$-surface type at the positions marked with ``$A_i^{(1)}$", and there exist $j$\,th-order $q$-{\ODE}s at the positions marked with ``$j$\,th".
\item 
The number $k$ in the symbol ``$[k]$" represents the number of parameters of an appearing $q$-{\ODE} including an independent variable but excluding a shift parameter. 
For example, let us consider the system \eqref{eqn:intro_dP_odd} obtained from the $(2N+1,1)$-pair, where $N\geq 1$. 
The parameters involved in the system are $\{a_1,\dots,a_{2N},b,c,p\}$. 
Thus, the number of parameters, including the independent variable ``$b$" but excluding the shift parameter ``$p$", is $2N+2$.
\item 
Moving to the right in two steps ({\it i.e.}, $(n,L)\to (n+2,L)$) increases the order of appearing $q$-{\ODE}s by two and also increases the number of their parameters by two. 
On the other hand, moving downward in two steps ({\it i.e.}, $(n,L)\to (n,L+2)$) increases the order of appearing $q$-{\ODE}s by two, but the number of their parameters remains unchanged.
\item 
An appropriate limit operation induces degeneration in the bottom left direction of the table (that is, $(n,L)\to (n-1,L+1)$).
\end{itemize}

\begin{table}[htp]
\begin{center}
$\begin{array}{c||c|c|c|c|c|c|c}
\text{\diagbox{$L$}{$n$}}&1&2&3&4&5&6&\cdots\\
\hhline{=|tb|=|=|=|=|=|=|=}
0&\times&\times&A_5^{(1)} {}_{[3]}\hspace{-0.6em}&A_3^{(1)} {}_{[5]}\hspace{-0.6em}&\text{$4$th}~ {}_{[5]}\hspace{-0.6em}&\text{$4$th}~ {}_{[7]}\hspace{-0.6em}&\cdots\\
\hline
1&\times&A_6^{(1)} {}_{[2]}\hspace{-0.6em}&A_4^{(1)} {}_{[4]}\hspace{-0.6em}&\text{$4$th}~ {}_{[4]}\hspace{-0.6em}&\text{$4$th}~ {}_{[6]}\hspace{-0.6em}&\text{$6$th}~ {}_{[6]}\hspace{-0.6em}&\cdots\\
\hline
2&A_7^{(1)} {}_{[1]}\hspace{-0.6em}&A_5^{(1)} {}_{[3]}\hspace{-0.6em}&\text{$4$th}~ {}_{[3]}\hspace{-0.6em}&\text{$4$th}~ {}_{[5]}\hspace{-0.6em}&\text{$6$th}~ {}_{[5]}\hspace{-0.6em}&\text{$6$th}~ {}_{[7]}\hspace{-0.6em}&\cdots\\
\hline
3&A_6^{(1)} {}_{[2]}\hspace{-0.6em}&\text{$4$th}~ {}_{[2]}\hspace{-0.6em}&\text{$4$th}~ {}_{[4]}\hspace{-0.6em}&\text{$6$th}~ {}_{[4]}\hspace{-0.6em}&\text{$6$th}~ {}_{[6]}\hspace{-0.6em}&\text{$8$th}~ {}_{[6]}\hspace{-0.6em}&\cdots\\
\hline
4&\text{$4$th}~ {}_{[1]}\hspace{-0.6em}&\text{$4$th}~ {}_{[3]}\hspace{-0.6em}&\text{$6$th}~ {}_{[3]}\hspace{-0.6em}&\text{$6$th}~ {}_{[5]}\hspace{-0.6em}&\text{$8$th}~ {}_{[5]}\hspace{-0.6em}&\text{$8$th}~ {}_{[7]}\hspace{-0.6em}&\cdots\\
\hline
5&\text{$4$th}~ {}_{[2]}\hspace{-0.6em}&\text{$6$th}~ {}_{[2]}\hspace{-0.6em}&\text{$6$th}~ {}_{[4]}\hspace{-0.6em}&\text{$8$th}~ {}_{[4]}\hspace{-0.6em}&\text{$8$th}~ {}_{[6]}\hspace{-0.6em}&\text{$10$th}\, {}_{[6]}\hspace{-0.6em}&\cdots\\
\hline
\vdots&\vdots&\vdots&\vdots&\vdots&\vdots&\vdots&\ddots
\end{array}$
\end{center}
\caption{Table presenting results and conjectures for the $(n, L)$-pair.}
\label{table:nL}
\end{table}

\subsection{A conjecture on linear problems for $q$-{\ODE}s obtained from the $(n,L)$-pair}
In this subsection, we describe a conjecture concerning a linear problem of Painlev\'e type $q$-{\ODE}s derived from the $(n, L)$-pair, where $L\geq1$. 
We provide an explanation for the pair \eqref{eqns:ABS_U} and \eqref{eqn:ABS_U_cond}, but the same applies to the pair \eqref{eqn:D4} and \eqref{eqn:D4_cond}.

Let $L\geq1$.
In \cite{nakazono2023higerA4A6}, a Lax representation of the system \eqref{eqns:ABS_U} is given by
\begin{subequations}\label{eqns:Weyl_phi}
\begin{align}
 &\phi_{\overline{i}}
 =\begin{pmatrix}
 \dfrac{\mu U_{\overline{i}}}{\al^{(i)}(l_i)U}&-\dfrac{{U_{\overline{i}}}^2}{\la(l_{0\cdots n})}\\[1em]
 \dfrac{\la(l_{0\cdots n})}{U^2}&-\dfrac{\mu U_{\overline{i}}}{\al^{(i)}(l_i)U}
 \end{pmatrix}.\phi,\quad
 i=1,\dots,n,\\
 &\phi_{\overline{0}}
 =\begin{pmatrix}
 -\dfrac{\mu\ka(l_0)U_{\overline{0}}}{U}&-\la(l_{0\cdots n})^2{U_{\overline{0}}}^2\\[1em]
 \dfrac{1}{\la(l_{0\cdots n})^2U^2}&0
 \end{pmatrix}.\phi,
\end{align}
\end{subequations}
where $\mu\in\bbC$ is a spectral variable, $\phi=\phi(l_1,\dots,l_n,l_0)$ is a column vector of length two and 
\begin{equation}
 \phi_{\overline{i}}=\phi|_{\,l_i\to l_i+1}.
\end{equation}
Indeed, we can easily verify that the compatibility conditions
\begin{equation}
 \Big(\phi_{\overline{i}}\Big)_{\overline{j}}=\Big(\phi_{\overline{j}}\Big)_{\overline{i}}~,\quad
 0\leq i<j\leq n,
\end{equation}
provide the system \eqref{eqns:ABS_U}.
Imposing condition \eqref{eqn:ABS_U_cond} into system \eqref{eqns:ABS_U}, we obtain the following conditions:
\begin{equation}
 \dfrac{\al^{(1)}(l_1)}{\al^{(1)}(l_1+1)}
 =\cdots
 =\dfrac{\al^{(n)}(l_{n})}{\al^{(n)}(l_{n}+1)}
 =\dfrac{\ka(l_0+L)}{\ka(l_0)}
 =q,
\end{equation}
where $q\in\bbC$.
Note that if $(n+L)$ is odd, then the following additional condition is required.
\begin{equation}
 \la_0=1.
\end{equation}
Define the shift operators $T_i$, $i=0,\dots,n$, 
by the actions on the parameters $\{\al^{(1)}(l),\dots,\al^{(n)}(l),\ka(l),\la_0\}_{l\in\bbZ}$ and $\mu$, the variable $U=U(l_1,\dots,l_n,l_0)$ and the vector $\phi=\phi(l_1,\dots,l_n,l_0)$ as follows:
\begin{subequations}
\begin{align}
 &T_i(\al^{(j)}(l))
 =\begin{cases}
 \al^{(i)}(l+1)&\text{if } j=i,\\
 \al^{(j)}(l)&\text{otherwise},
 \end{cases}\quad
 T_i(\ka(l))
 =\begin{cases}
 \ka(l+1)&\text{if } i=0,\\
 \ka(l)&\text{otherwise},
 \end{cases}\\
 &T_i(\la_0)=\dfrac{1}{\la_0},\quad
 T_i(\mu)=\mu,\quad
 T_i(U)=U_{\overline{i}},\quad
 T_i(\phi)=\phi_{\overline{i}}.
\end{align}
\end{subequations}
Moreover, define the parameter $x$ and the operator $T_x$ by
\begin{equation}
 x=\dfrac{\mu}{\prod_{k=1}^n\al^{(k)}(0)^{1/n}},\quad
 T_x=T_1T_2\cdots T_n {T_0}^L.
\end{equation}
Then, $x$ and $T_x$ can be regarded as the spectral variable and spectral operator of a Lax pair of a Painlev\'e type $q$-{\ODE} obtained from the $(n,L)$-pair, respectively.
The operator $T_x$ acts on $x$ as follows:
\begin{equation}
T_x(x)=qx.
\end{equation}
On the other hand, the Painlev\'e parameters (an independent variable and parameters of a Painlev\'e type $q$-{\ODE} derived from the $(n,L)$-pair) and the Painlev\'e variable (a dependent variable of the Painlev\'e type $q$-{\ODE}) are invariant under the action of $T_x$.
In general, the Painlev\'e parameters are expressed as rational functions of the parameters $\{\al^{(1)}(l),\dots,\al^{(n)}(l),\ka(l),\la_0\}_{l\in\bbZ}$ (or their rational number powers), and the Painlev\'e variable is expressed as a rational function of the $U$-variable.
A spectral direction of a linear problem for Painlev\'e type $q$-{\ODE}s derived from the $(n,L)$-pair takes the following form:
\begin{equation}
 T_x(\phi)
 =\begin{pmatrix}\ast \,x&\ast\\\ast&0\end{pmatrix}
 \dots
 \begin{pmatrix}\ast \,x&\ast\\\ast&0\end{pmatrix}.
 \begin{pmatrix}\ast \,x&\ast\\\ast&\ast \,x\end{pmatrix}
 \dots
 \begin{pmatrix}\ast \,x&\ast\\\ast&\ast \,x\end{pmatrix}.
 \phi,
\end{equation}
where the coefficient matrices above are given by the products of $L$ matrices of the form:
\begin{equation}
 \begin{pmatrix}\ast \,x&\ast\\\ast&0\end{pmatrix}
\end{equation}
and $n$ matrices of the form
\begin{equation}
 \begin{pmatrix}\ast \,x&\ast\\\ast&\ast \,x\end{pmatrix}.
\end{equation}
Here, the entries denoted by ``$\ast$" are expressed as rational functions of the parameters $\{\al^{(1)}(l),\dots,\al^{(n)}(l),\ka(l),\la_0\}_{l\in\bbZ}$ and the $U$-variables, and the entries remain invariant under the action of $T_x$. 
Expressing the ``$\ast$" entries solely in terms of the Painlev\'e parameters and variables may require careful gauge transformations of vector $\phi$. 
For specific examples, see \cite{nakazono2023higerA4A6} for the $(n,1)$-pair. 
(See also \cite{JNS2016:MR3584386} for the $(3,1)$-pair.)

\subsection{Painlev\'e type $q$-{\ODE}s obtained from the $(1,L)$-pair}\label{subequation:(1,L)}
In this subsection, we show that the $(1,L)$-pair gives the Painlev\'e type $q$-{\ODE}s in \cite{Okubo2017:arxiv1704.05403,MOT2021:Cluster}.
We consider the cases separately based on the parity of $L$.

\subsubsection{The case $L$ is even}
Consider the $(1,2h)$-pair, where $h\in\bbZ_{\geq1}$.
In this case, condition \eqref{eqn:D4_cond} is equivalent to expressing $U(l_1,l_0)$ as:
\begin{equation}\label{eqn:D4_cond_2}
 U(l_1,l_0)=U(l_0-2hl_1).
\end{equation}
Substituting \eqref{eqn:D4_cond_2} into Equation \eqref{eqn:D4}, we obtain
\begin{equation}\label{eqn:12h_U}
 \dfrac{U(l+2h+1)}{U(l)}+\dfrac{U(l+1)}{U(l+2h)}+p \al^{(1)}(0) \ka(l)=0,
\end{equation}
with the conditions of the parameters
\begin{equation}
 \dfrac{\al(l_1)}{\al(l_1+1)}=\dfrac{\ka(l_0+2h)}{\ka(l_0)}=p,\quad
 \la_0=1,
\end{equation}
where $p\in\bbC$.
Setting
\begin{equation}
 x_l=\dfrac{U(l+1)}{U(l)},
\end{equation}
from \eqref{eqn:12h_U} we obtain
\begin{equation}\label{eqn:2h_x}
 x_{l+2h}x_l
 =-\dfrac{1}{\prod_{k=1}^{2h-1}x_{l+k}}\left(\dfrac{1}{\prod_{k=1}^{2h-1}x_{l+k}}+p \al^{(1)}(0) \ka(l)\right).
\end{equation}
Set
\begin{equation}
 y_l
 =p \al^{(1)}(0) \ka(l)\dfrac{U(l+2h)}{U(l+1)}
 =p \al^{(1)}(0) \ka(l)\prod_{k=1}^{2h-1}x_{l+k}.
\end{equation}
Then, from Equation \eqref{eqn:2h_x} the following relation holds:
\begin{equation}
 x_{l+2h}x_l=-p^2 {\al^{(1)}(0)}^2 {\ka(l)}^2\dfrac{y_l+1}{{y_l}^2}.
\end{equation}
Therefore, we obtain
\begin{align}
 y_{l+2h}y_l
 &=p^3 {\al^{(1)}(0)}^2 {\ka(l)}^2\prod_{k=1}^{2h-1}x_{l+2h+k}x_{l+k}\notag\\
 &=p^3 {\al^{(1)}(0)}^2 {\ka(l)}^2\left(\prod_{k=1}^{2h-1}-p^2 {\al^{(1)}(0)}^2 {\ka(l+k)}^2\dfrac{y_{l+k}+1}{{y_{l+k}}^2}\right)\notag\\
 &=p^{4h+3} {\al^{(1)}(0)}^{4h+2} \left(\prod_{k=0}^{2h-1} {\ka(l+k)}^2\right)\left(\prod_{k=1}^{2h-1}\dfrac{y_{l+k}+1}{{y_{l+k}}^2}\right).
\end{align}
Finally, defining $q$ and $t_l$ by
\begin{equation}
 q=p^2,\quad
 t_l=q^{4h+3} {\al^{(1)}(0)}^{4h+2} \left(\prod_{k=0}^{2h-1} {\ka(l+k)}^2\right),
\end{equation}
we obtain the following one kind of Painlev\'e type $q$-{\ODE}s in \cite{Okubo2017:arxiv1704.05403,MOT2021:Cluster}:
\begin{equation}\label{eqn:Okubo_1}
 y_{l+2h}y_l=t_l\left(\prod_{k=1}^{2h-1}\dfrac{y_{l+k}+1}{{y_{l+k}}^2}\right),
\end{equation}
where $t_l=q^lt_0$.

\begin{remark}
Equation \eqref{eqn:Okubo_1} has an extended affine Weyl group symmetry of $A_1^{(1)}$-type (see \cite{MOT2021:Cluster} for details).
\end{remark}

\begin{remark}
When $h=1$, Equation \eqref{eqn:Okubo_1} is equivalent to a $q$-Painlev\'e I equation of $A_7^{(1)}$-surface type\cite{GR2000:zbMATH01498348,SakaiH2001:MR1882403} (see \cite{Okubo2017:arxiv1704.05403} for details).
\end{remark}

\subsubsection{The case $L$ is odd}
Consider the $(1,2h+1)$-pair, where $h\in\bbZ_{\geq1}$.
In this case, condition \eqref{eqn:D4_cond} is equivalent to expressing $U(l_1,l_0)$ as:
\begin{equation}\label{eqn:D4_cond_3}
 U(l_1,l_0)=U(l_0-(2h+1)l_1).
\end{equation}
Substituting \eqref{eqn:D4_cond_3} into Equation \eqref{eqn:D4}, we obtain
\begin{equation}\label{eqn:12h+1_U}
 \dfrac{U(l+2h+2)}{U(l)}+\la(l)^4\dfrac{U(l+1)}{U(l+2h+1)}+q\al^{(1)}(0) \ka(l)\la(l)^3=0,
\end{equation}
with the conditions of the parameters
\begin{equation}
 \dfrac{\al(l_1)}{\al(l_1+1)}=\dfrac{\ka(l_0+2h+1)}{\ka(l_0)}=q,
\end{equation}
where $q\in\bbC$.
Setting
\begin{equation}
 x_l=\dfrac{U(l+2)}{U(l)},
\end{equation}
from \eqref{eqn:12h+1_U} we obtain
\begin{equation}\label{eqn:2h+1_x}
 x_{l+2h}x_l
 =-\dfrac{\la(l)^3}{\prod_{k=1}^{h-1}x_{l+2k}}\left(\dfrac{\la(l)}{\prod_{k=0}^{h-1}x_{l+2k+1}}+q\al^{(1)}(0) \ka(l)\right).
\end{equation}
Set
\begin{equation}
 y_l
 =q\al^{(1)}(0) \ka(l-1)\la(l)\dfrac{U(l+2h)}{U(l)}
 =q\al^{(1)}(0) \ka(l-1)\la(l)\prod_{k=0}^{h-1}x_{l+2k}.
\end{equation}
Then, from Equation \eqref{eqn:2h+1_x} the following relation holds:
\begin{equation}
 x_{l+2h}=-q^2{\al^{(1)}(0)}^2 \ka(l-1)\ka(l)\la(l)^4\dfrac{y_{l+1}+1}{y_ly_{l+1}}.
\end{equation}
Therefore, we obtain
\begin{align}
 y_{l+2h}
 &=q\al^{(1)}(0) \ka(l+2h-1)\la(l)\prod_{k=0}^{h-1}x_{l+2h+2k}\notag\\
 &=(-1)^h q^{2h}{\al^{(1)}(0)}^{2h+1}\la(l)^{4h+1}
 \left(\prod_{k=0}^{2h}\ka(l+k)\right)
 \left(\prod_{k=0}^{h-1}\dfrac{y_{l+2k+1}+1}{y_{l+2k}y_{l+2k+1}}\right)\notag\\
 &=(-1)^h q^{2h}{\al^{(1)}(0)}^{2h+1}\la(l)^{4h+1}
 \left(\prod_{k=0}^{2h}\ka(l+k)\right)
 \dfrac{\prod_{k=0}^{h-1}\left(y_{l+2k+1}+1\right)}{y_l\left(\prod_{k=1}^{2h-1}y_{l+k}\right)}.
\end{align}
Finally, defining $c$ and $t_l$ by
\begin{equation}
 c={\la_0}^{4h+1},\quad
 t_l=(-1)^h q^{2h}{\al^{(1)}(0)}^{2h+1}\left(\prod_{k=0}^{2h}\ka(l+k)\right),
\end{equation}
we obtain the following another kind of Painlev\'e type $q$-{\ODE}s in \cite{Okubo2017:arxiv1704.05403,MOT2021:Cluster}:
\begin{equation}\label{eqn:Okubo_2}
 y_{l+2h}y_l=c^{(-1)^l}t_l\dfrac{\prod_{k=0}^{h-1}\left(y_{l+2k+1}+1\right)}{\prod_{k=1}^{2h-1}y_{l+k}},
\end{equation}
where $t_l=q^lt_0$.

\begin{remark}
Equation \eqref{eqn:Okubo_2} has an extended affine Weyl group symmetry of $(A_1\times A_1)^{(1)}$-type (see \cite{MOT2021:Cluster} for details).
\end{remark}

\begin{remark}
When $h=1$, Equation \eqref{eqn:Okubo_2} is equivalent to a $q$-Painlev\'e II equation of $A_6^{(1)}$-surface type\cite{RG1996:MR1399286,KTGR2000:MR1789477,RGTT2001:MR1838017,SakaiH2001:MR1882403} (see \cite{Okubo2017:arxiv1704.05403} for details).
\end{remark}

\subsection{Degeneration of the $(3, L)$- and $(2, L)$-pairs through limit operations}\label{subsection:(3,L)(2,L)_reductions}
In this subsection, we demonstrate degeneration from the $(3, L)$-pair, where $L\geq1$, to the $(2, L+1)$-pair and from the $(2, L)$-pair, where $L\geq1$, to the $(1, L+1)$-pair through limit operations.

\subsubsection{Degeneration from the $(3, L)$-pair to the $(2, L+1)$-pair}
Let $L\geq1$.
Then, the $(3, L)$-pair is given by the system of {\PDE}s
\begin{subequations}\label{eqns:ABS_U_3L}
\begin{align}
 &\la(l_{0\cdots 3})^2\dfrac{U_{\overline{12}}}{U}=\dfrac{\al^{(1)}(l_1) U_{\overline{1}}-\al^{(2)}(l_2) U_{\overline{2}}}{\al^{(2)}(l_2) U_{\overline{1}}-\al^{(1)}(l_1) U_{\overline{2}}},\label{eqn:ABS_U_12}\\
 &\la(l_{0\cdots 3})^2\dfrac{U_{\overline{13}}}{U}=\dfrac{\al^{(1)}(l_1) U_{\overline{1}}-\al^{(3)}(l_3) U_{\overline{3}}}{\al^{(3)}(l_3) U_{\overline{1}}-\al^{(1)}(l_1) U_{\overline{3}}},\label{eqn:ABS_U_13}\\
 &\la(l_{0\cdots 3})^2\dfrac{U_{\overline{23}}}{U}=\dfrac{\al^{(2)}(l_2) U_{\overline{2}}-\al^{(3)}(l_3) U_{\overline{3}}}{\al^{(3)}(l_3) U_{\overline{2}}-\al^{(2)}(l_2) U_{\overline{3}}},\label{eqn:ABS_U_23}\\
 &\dfrac{U_{\overline{0\,1}}}{U}+\la(l_{0\cdots 3})^4\dfrac{U_{\overline{0}}}{U_{\overline{1}}}+\al^{(1)}(l_1) \ka(l_0) \la(l_{0\cdots 3})=0,\label{eqn:ABS_U_01}\\
 &\dfrac{U_{\overline{0\,2}}}{U}+\la(l_{0\cdots 3})^4\dfrac{U_{\overline{0}}}{U_{\overline{2}}}+\al^{(2)}(l_2) \ka(l_0) \la(l_{0\cdots 3})=0,\label{eqn:ABS_U_02}\\
 &\dfrac{U_{\overline{0\,3}}}{U}+\la(l_{0\cdots 3})^4\dfrac{U_{\overline{0}}}{U_{\overline{3}}}+\al^{(3)}(l_3) \ka(l_0) \la(l_{0\cdots 3})=0,\label{eqn:ABS_U_03}
\end{align}
\end{subequations}
where $U=U(l_1,l_2,l_3,l_0)$, and the periodic condition
 \begin{equation}\label{eqn:ABS_U_cond_3L}
  U(l_1+1,l_2+1,l_3+1,l_0+L)=U(l_1,l_2,l_3,l_0).
\end{equation}
Due to the periodic condition \eqref{eqn:ABS_U_cond_3L}, the parameters satisfy the following conditions:
\begin{equation}
 \dfrac{\al^{(1)}(l_1)}{\al^{(1)}(l_1+1)}
 =\dfrac{\al^{(2)}(l_2)}{\al^{(2)}(l_2+1)}
 =\dfrac{\al^{(3)}(l_3)}{\al^{(3)}(l_3+1)}
 =\dfrac{\ka(l_0+L)}{\ka(l_0)}
 =q,
\end{equation}
where $q\in\bbC$, which give
\begin{equation}
 \al^{(1)}(l_1)=q^{-l_1}\al^{(1)}(0),\quad
 \al^{(2)}(l_2)=q^{-l_2}\al^{(2)}(0),\quad
 \al^{(3)}(l_3)=q^{-l_3}\al^{(3)}(0).
\end{equation}
If $L$ is even, then the following additional condition is needed:
\begin{equation}
 \la_0=1.
\end{equation}
Therefore, regardless of the parity of $L$, the following relation holds:
\begin{equation}
 \la(l+L)=\dfrac{1}{\la(l)}.
\end{equation}

In the following, we show that pair \eqref{eqns:ABS_U_3L} and \eqref{eqn:ABS_U_cond_3L} is reduced to the $(2, L+1)$-pair through a limit operation.
Condition \eqref{eqn:ABS_U_cond_3L} is equivalent to expressing $U(l_1,l_2,l_3,l_0)$ as
\begin{equation}
 U(l_1,l_2,l_3,l_0)=\omega(l_1-l_3,l_2-l_3,l_0-Ll_3).
\end{equation}
Set
\begin{subequations}
\begin{align}
 &\omega(l_1,l_2,l_0)=\xi^{-l_0}A(l_0)\la(l_0+l_1+l_2)^{\frac{3(2l_0-l_1-l_2)}{2L}} u(l_1,l_2,l_0),\label{eqn:limit_Uu}\\
 &\al^{(1)}(0)=\xi^L a^{(1)}(0),\quad
 \al^{(2)}(0)=\xi^L a^{(2)}(0),\quad
 \ka(l_0)=\xi^{-L-1} K(l_0),
\end{align}
\end{subequations}
where $\xi$ is a positive real number and $A(l)\in\bbC$ is a function satisfying
\begin{equation}
 A(l+L+1)=-qK(l) \al^{(3)}(0) A(l).
\end{equation}
Then, taking $\xi\to 0$, we obtain a pair of the system of {\PDE}s
\begin{subequations}\label{eqns:ABS_u_2L+1}
\begin{align}
 &\La(l_0+l_1+l_2)^2\dfrac{u_{\overline{12}}}{u}=\dfrac{a^{(1)}(l_1) u_{\overline{1}}-a^{(2)}(l_2) u_{\overline{2}}}{a^{(2)}(l_2) u_{\overline{1}}-a^{(1)}(l_1) u_{\overline{2}}},\label{eqn:ABS_u_12}\\
 &\dfrac{u_{\overline{0\,1}}}{u}+\La(l_0+l_1+l_2)^4\dfrac{u_{\overline{0}}}{u_{\overline{1}}}+a^{(1)}(l_1) k(l_0) \La(l_0+l_1+l_2)=0,\label{eqn:ABS_u_01}\\
 &\dfrac{u_{\overline{0\,2}}}{u}+\La(l_0+l_1+l_2)^4\dfrac{u_{\overline{0}}}{u_{\overline{2}}}+a^{(2)}(l_2) k(l_0) \La(l_0+l_1+l_2)=0,\label{eqn:ABS_u_02}
\end{align}
\end{subequations}
where 
\begin{equation}
 u=u(l_1,l_2,l_0),\quad
 k(l):=\dfrac{A(l)K(l)}{A(l+1)},\quad
 \La(l):=\la(l)^{\frac{2L-3}{2L}},
\end{equation}
and the periodic condition
 \begin{equation}\label{eqn:ABS_u_cond_2L+1}
  u(l_1+1,l_2+1,l_0+L+1)=u(l_1,l_2,l_0),
\end{equation}
which is equivalent to the $(2,L+1)$-pair.
Indeed, Equation \eqref{eqn:ABS_U_12} is reduced to Equation \eqref{eqn:ABS_u_12},
Equation \eqref{eqn:ABS_U_01} is reduced to Equation \eqref{eqn:ABS_u_01},
Equation \eqref{eqn:ABS_U_02} is reduced to Equation \eqref{eqn:ABS_u_02}, 
and Equation \eqref{eqn:ABS_U_03} is reduced to condition \eqref{eqn:ABS_u_cond_2L+1}.
Moreover, the reduced equations obtained from Equations \eqref{eqn:ABS_U_13} and \eqref{eqn:ABS_U_23} can be derived from 
Equation \eqref{eqn:ABS_u_01} with condition \eqref{eqn:ABS_u_cond_2L+1}
and Equations \eqref{eqn:ABS_u_02} with condition \eqref{eqn:ABS_u_cond_2L+1}, respectively.
Note that the parameters in system \eqref{eqns:ABS_u_2L+1} satisfy the following relations:
\begin{equation}
 \dfrac{a^{(1)}(l_1)}{a^{(1)}(l_1+1)}
 =\dfrac{a^{(2)}(l_2)}{a^{(2)}(l_2+1)}
 =\dfrac{k(l_0+L+1)}{k(l_0)}
 =q,\quad
 \La(l+1)=\dfrac{1}{\La(l)}.
\end{equation}

\subsubsection{Degeneration from the $(2, L)$-pair to the $(1, L+1)$-pair}
Let $L\geq1$.
Then, the $(2, L)$-pair is given by the system of {\PDE}s
\begin{subequations}\label{eqns:ABS_U_2L}
\begin{align}
 &\la(l_{0\cdots 2})^2\dfrac{U_{\overline{12}}}{U}=\dfrac{\al^{(1)}(l_1) U_{\overline{1}}-\al^{(2)}(l_2) U_{\overline{2}}}{\al^{(2)}(l_2) U_{\overline{1}}-\al^{(1)}(l_1) U_{\overline{2}}},\label{eqn:ABS_U_12_2L}\\
 &\dfrac{U_{\overline{0\,1}}}{U}+\la(l_{0\cdots 2})^4\dfrac{U_{\overline{0}}}{U_{\overline{1}}}+\al^{(1)}(l_1) \ka(l_0) \la(l_{0\cdots 2})=0,\label{eqn:ABS_U_01_2L}\\
 &\dfrac{U_{\overline{0\,2}}}{U}+\la(l_{0\cdots 2})^4\dfrac{U_{\overline{0}}}{U_{\overline{2}}}+\al^{(2)}(l_2) \ka(l_0) \la(l_{0\cdots 2})=0,\label{eqn:ABS_U_02_2L}
\end{align}
\end{subequations}
where $U=U(l_1,l_2,l_0)$, and the periodic condition
 \begin{equation}\label{eqn:ABS_U_cond_2L}
  U(l_1+1,l_2+1,l_0+L)=U(l_1,l_2,l_0).
\end{equation}
Due to the periodic condition \eqref{eqn:ABS_U_cond_2L}, the parameters satisfy the following conditions:
\begin{equation}
 \dfrac{\al^{(1)}(l_1)}{\al^{(1)}(l_1+1)}
 =\dfrac{\al^{(2)}(l_2)}{\al^{(2)}(l_2+1)}
 =\dfrac{\ka(l_0+L)}{\ka(l_0)}
 =q,
\end{equation}
where $q\in\bbC$, which give
\begin{equation}
 \al^{(1)}(l_1)=q^{-l_1}\al^{(1)}(0),\quad
 \al^{(2)}(l_2)=q^{-l_2}\al^{(2)}(0).
\end{equation}
If $L$ is odd, then the following additional condition is needed:
\begin{equation}
 \la_0=1.
\end{equation}
Therefore, regardless of the parity of $L$, the following relation holds:
\begin{equation}
 \la(l+L)=\la(l).
\end{equation}

In the following, we show that pair \eqref{eqns:ABS_U_2L} and \eqref{eqn:ABS_U_cond_2L} is reduced to the $(1, L+1)$-pair through a limit operation.
Condition \eqref{eqn:ABS_U_cond_2L} is equivalent to expressing $U(l_1,l_2,l_0)$ as:
\begin{equation}
 U(l_1,l_2,l_0)=\omega(l_1-l_2,l_0-Ll_2).
\end{equation}
Set
\begin{subequations}
\begin{align}
 &\omega(l_1,l_0)=\xi^{-l_0}A(l_0)\la(l_0+l_1)^{\frac{3(2l_0-l_1)}{2L+1}} u(l_1,l_0),\\
 &\al^{(1)}(0)=\xi^L a^{(1)}(0),\quad
 \ka(l_0)=\xi^{-L-1} K(l_0),
\end{align}
\end{subequations}
where $\xi$ is a positive real number and $A(l)\in\bbC$ is a function satisfying
\begin{equation}
 A(l+L+1)=-qK(l) \al^{(2)}(0) A(l).
\end{equation}
Then, taking $\xi\to 0$, we obtain the following pair:
\begin{align}
 &\dfrac{u_{\overline{0\,1}}}{u}+\La(l_0+l_1)^4\dfrac{u_{\overline{0}}}{u_{\overline{1}}}+a^{(1)}(l_1) k(l_0) \La(l_0+l_1)=0,\label{eqn:red_2L_eqn}\\
 &u(l_1+1,l_0+L+1)=u(l_1,l_0),\label{eqn:red_2L_cond}
\end{align}
where 
\begin{equation}
 u=u(l_1,l_0),\quad
 k(l):=\dfrac{A(l)K(l)}{A(l+1)},\quad
 \La(l):=\la(l)^{\frac{2(L-1)}{2L+1}},
\end{equation}
which is equivalent to the $(1,L+1)$-pair.
Indeed, Equations \eqref{eqn:ABS_U_01_2L} and  \eqref{eqn:ABS_U_02_2L} are reduced to Equations \eqref{eqn:red_2L_eqn} and \eqref{eqn:red_2L_cond}, respectively.
Moreover, the reduced equation obtained from Equation \eqref{eqn:ABS_U_12_2L} can be derived from 
Equation \eqref{eqn:red_2L_eqn} under condition \eqref{eqn:red_2L_cond}.
Note that the parameters in Equation \eqref{eqn:red_2L_eqn} satisfy the following relations:
\begin{equation}
 \dfrac{a^{(1)}(l_1)}{a^{(1)}(l_1+1)}
 =\dfrac{k(l_0+L+1)}{k(l_0)}
 =q,\quad
 \La(l+1)=\dfrac{1}{\La(l)}.
\end{equation}
\def\cprime{$'$} \def\cprime{$'$}

\end{document}